\documentclass[11pt,onecolumn]{article}

\RequirePackage[utf8]{inputenc}
\RequirePackage[T1]{fontenc}
\usepackage{amsmath,amsthm,amssymb,amsfonts,mathtools,epsfig,graphicx}
\RequirePackage{setspace}
\AtBeginDocument{%
  \setlength{\abovedisplayskip}{5pt plus 2pt minus 2pt}%
  \setlength{\belowdisplayskip}{5pt plus 2pt minus 2pt}%
  \setlength{\abovedisplayshortskip}{3pt plus 2pt minus 2pt}%
  \setlength{\belowdisplayshortskip}{3pt plus 2pt minus 2pt}%
}
\usepackage{mathpazo}
\usepackage{tgpagella}
\usepackage[tracking=true,expansion=true,protrusion=true]{microtype}
\usepackage[linesnumbered,ruled,vlined]{algorithm2e}
\usepackage{dsfont}
\usepackage{enumerate}
\usepackage{url}
\usepackage{bm}
\usepackage{thmtools}
\usepackage{enumitem}
\usepackage{mathbbol}
\RequirePackage[
  a4paper,
  textheight=9in,
  textwidth=6.5in,
  top=1in,
  headheight=14pt,
  headsep=25pt,
  footskip=30pt,
]{geometry}

\usepackage{float}
\usepackage{booktabs}
\usepackage{array}
\newcolumntype{L}[1]{>{\raggedright\arraybackslash}p{#1}}
\newcolumntype{C}[1]{>{\centering\arraybackslash}p{#1}}
\RequirePackage{xcolor}
\definecolor{Black}{RGB}{0, 0, 0}
\definecolor{White}{RGB}{255, 255, 255}
\definecolor{Gray}{RGB}{127, 127, 127}
\definecolor{LightGray}{RGB}{216, 216, 216}
\definecolor{Orange}{RGB}{237, 125, 49}
\definecolor{LightOrange}{RGB}{251, 229, 214}
\definecolor{Yellow}{RGB}{255, 192, 0}
\definecolor{LightYellow}{RGB}{255, 242, 200}
\definecolor{Blue}{RGB}{0, 0, 255}
\definecolor{LightBlue}{RGB}{222, 235, 247}
\definecolor{Green}{RGB}{0, 128, 0}
\definecolor{LightGreen}{RGB}{226, 240, 217}
\definecolor{Navy}{RGB}{30, 70, 235}
\definecolor{LightNavy}{RGB}{218, 227, 243}
\RequirePackage[
  plainpages=false,
  colorlinks=true,
  linkcolor=Navy,
  citecolor=Green,
  urlcolor=Green,
]{hyperref}
\RequirePackage[
  backend=biber,
  style=alphabetic,
  url=false,
  doi=false,
  isbn=false,
  maxalphanames=4,
  minalphanames=3,
  maxnames=99,
  backref=true,
]{biblatex}
\DeclareSourcemap{
  \maps{
    \map{
      \step[fieldset=doi, null]
      \step[fieldset=url, null]
      \step[fieldset=isbn, null]
      \step[fieldset=editor, null]
    }
  }
}
\renewcommand*{\bibpagerefpunct}{\addperiod\space}
\renewbibmacro*{pageref}{%
  \iflistundef{pageref}
    {}
    {\printlist[pageref][-\value{listtotal}]{pageref}\nopunct}}
\renewbibmacro*{setpageref}{%
  \iflistundef{pageref}
    {}
    {\setunit{\bibpagerefpunct}\newblock
     \printlist[pageref][-\value{listtotal}]{pageref}\nopunct}}
\defbibheading{bibliography}[\refname]{%
  \ifcsundef{chapter}
    {\section*{#1}\addcontentsline{toc}{section}{#1}}
    {\chapter*{#1}\addcontentsline{toc}{chapter}{#1}}%
  \markboth{#1}{#1}}
  
\RequirePackage[nameinlink]{cleveref}
\RequirePackage{tcolorbox}

\let\oldnl\nl
\newcommand{\nonl}{\renewcommand{\nl}{\let\nl\oldnl}}

\newcommand{\RNC}{\ensuremath{\mathbf{RNC}}}
\newcommand{\NC}{\ensuremath{\mathbf{NC}}}
\makeatletter
\providecommand{\rep@title}{}
\newtheorem*{rep@theorem}{\rep@title}
\newcommand{\newreptheorem}[2]{%
\newenvironment{rep#1}[1]{%
 \def\rep@title{#2 \ref{##1}}%
 \begin{rep@theorem}}%
 {\end{rep@theorem}}}
\makeatother

\newtheorem{theorem}{Theorem}
\newtheorem{corollary}[theorem]{Corollary}

\newtheorem{lemma}[theorem]{Lemma}
\newtheorem{definition}[theorem]{Definition}
\newtheorem{proposition}[theorem]{Proposition}

\newtheorem{fact}[theorem]{Fact}
\newreptheorem{proposition}{Proposition}
\newreptheorem{theorem}{Theorem}

\crefname{algocfline}{algorithm}{algorithms}
\Crefname{algocfline}{Algorithm}{Algorithms}
\crefname{table}{Table}{Tables}
\Crefname{table}{Table}{Tables}
\crefname{algorithm}{Algorithm}{Algorithms}
\Crefname{algorithm}{Algorithm}{Algorithms}
\crefname{algocf}{Algorithm}{Algorithms}
\Crefname{algocf}{Algorithm}{Algorithms}
\crefname{AlgoLine}{Line}{Lines}
\Crefname{AlgoLine}{Line}{Lines}

\DeclareMathOperator*{\polylog}{polylog}

\DeclareMathOperator*{\len}{len}
\DeclareMathOperator*{\maxwidth}{maxw}

\newcommand{\norm}[1]{\left\Vert#1\right\Vert}
\newcommand{\abs}[1]{\left\vert#1\right\vert}
\newcommand{\set}[1]{\left\{#1\right\}}
\newcommand{\tp}[1]{\left(#1\right)}

\begin{document}
\renewcommand{\d}[1]{\ensuremath{\operatorname{d}\!{#1}}}
\def\myparagraph#1{\vspace{2pt}\noindent{ #1~~}}
\newcommand{\eqdef}{{\stackrel{\mbox{\tiny \tt ~def~}}{=}}}

\long\def\ignore#1{}
\def\myps[#1]#2{\includegraphics[#1]{#2}}
\def\etal{{\em et al.}}
\def\Bar#1{{\bar #1}}
\def\br(#1,#2){{\langle #1,#2 \rangle}}
\def\setZ[#1,#2]{{[ #1 .. #2 ]}}
\def\Pr#1{{{\mathbb P}\!\bm{\left[} #1 \bm{\right]}}}

\def\maxsuperscript[#1]{\raisebox{#1}{\scriptsize $({\max})$}}

\def\bmin{{\beta_{\min}}}
\def\bmax{{\beta_{\max}}}
\def\betamin{\beta_{\min}}
\def\betamax{\beta_{\max}}

\def\bP{{\bf P}}

\def\E{{\mathbb E}}
\def\Var{{\mathbb V}\text{ar}}
\def\P{{\mathbb P}}
\def\X{{\mathbb X}}
\def\Y{{\mathbb Y}}
\def\Z{{\mathbb Z}}
\def\W{{\mathbb W}}
\def\U{{\mathbb U}}
\def\S{{\mathbb S}}
\def\R{{\mathbb R}}
\def\Vrel(#1){{{\mathbb S}[{#1}]}}

\def\Pcoef#1{{P^{#1}_{\tt count}}}
\def\Ptcoef#1{{P^{#1}_{\tt count-weak}}}
\def\Pratio#1{{P^{{#1}}_{\tt ratio-point}}}
\def\PratioAll#1{{P^{{#1}}_{\tt ratio-all}}}

\def\setof#1{{\left\{#1\right\}}}
\def\suchthat#1#2{\setof{\,#1\mid#2\,}}
\def\event#1{\setof{#1}}
\def\q={\quad=\quad}
\def\qq={\qquad=\qquad}
\def\ee{{\mathrm e}}
\def\calA{{\cal A}}
\def\calB{{\cal B}}
\def\calC{{\cal C}}
\def\calD{{\cal D}}
\def\calE{{\cal E}}
\def\calF{{\cal F}}
\def\calG{{\cal G}}
\def\calI{{\cal I}}
\def\calJ{{\cal J}}
\def\calH{{\cal H}}
\def\calHfrac{{{\cal L}}}
\def\calL{{\cal L}}
\def\calM{{\cal M}}
\def\calN{{\cal N}}
\def\calK{{\cal K}}
\def\calP{{\cal P}}
\def\calQ{{\cal Q}}
\def\calR{{\cal R}}
\def\calS{{\cal S}}
\def\A{\mathfrak A}
\def\calT{{\cal T}}
\def\calU{{\cal U}}
\def\calV{{\cal V}}
\def\calO{{\cal O}}
\def\calX{{\cal X}}
\def\calY{{\cal Y}}
\def\calZ{{\cal Z}}
\def\psfile[#1]#2{}
\def\psfilehere[#1]#2{}

\newcommand{\eps}{\varepsilon}
\newcommand\myqed{{}}

\title{Near-Optimal Parallel Approximate Counting via Sampling}

\author{%
\makebox[0pt][c]{%
\resizebox{0.96\textwidth}{!}{%
David G. Harris\thanks{Department of Computer Science, University of Maryland, College Park, USA. \texttt{davidgharris29@gmail.com}} \quad
Vladimir Kolmogorov\thanks{Institute of Science and Technology Austria, Klosterneuburg, Austria. \texttt{vnk@ist.ac.at}} \quad
Hongyang Liu\thanks{School of Computer Science, State Key Laboratory for Novel Software Technology, New Cornerstone Science Laboratory, Nanjing University, Nanjing, China. \texttt{\{liuhongyang,\,zhangyiyao\}@smail.nju.edu.cn}, \texttt{yinyt@nju.edu.cn}} \quad
Yitong Yin\footnotemark[\value{footnote}] \quad
Yiyao Zhang\footnotemark[\value{footnote}]%
}}%
}

\date{}
\maketitle

\begin{abstract}
  The computational equivalence between approximate counting and sampling is well established for polynomial-time algorithms.
  The most efficient general reduction from counting to sampling is based on simulated annealing.
  In this approach, the counting problem is formulated as estimating the ratio
  $Q={Z(\betamax)}/{Z(\betamin)}$ between partition functions
  $Z(\beta)=\sum_{x\in \Omega} \exp(\beta H(x))$
  of Gibbs distributions $\mu_\beta$ over $\Omega$ with Hamiltonian $H$,
  given access to a sampling oracle for $\mu_\beta$ at any $\beta \in [\betamin, \betamax]$.
  The sample complexity (measured by the number of oracle calls)
  is typically expressed in terms of $q$ and $h$,
  which respectively bound $\ln Q$ and $H$.
  The best upper bound achieved by known annealing algorithms with relative error $\eps$ is $O(q \eps^{-2} \log h)$.
  However, all known algorithms attaining this near-optimal complexity are inherently sequential, or \emph{adaptive}: the queried parameters $\beta$ depend on previous samples.

  We develop a simple non-adaptive algorithm for approximate counting using $O(q \eps^{-2} \log^2 h)$ samples,
  as well as an algorithm that achieves $O(q \eps^{-2} \log h)$ samples with just two additional adaptive rounds,
  matching the best sample complexity of sequential algorithms.
  These algorithms naturally yield work-efficient parallel (\textbf{RNC}) counting algorithms.
  We discuss applications to several classic models,
  including the anti-ferromagnetic 2-spin, monomer-dimer and ferromagnetic Ising models.
\end{abstract}

\section{Introduction}
A fundamental theme in randomized computation is the intrinsic connection between \emph{sampling} combinatorial objects and \emph{counting} those objects.
This interplay culminated in the Monte Carlo method and has become a central paradigm in algorithm design.
In theoretical computer science, a classic result of Jerrum, Valiant and Vazirani~\cite{jerrum_random_1986} established that, for all self-reducible problems, approximate counting and sampling are computationally equivalent up to polynomial time.

A wide range of counting problems can be formulated as computing the partition function of a Gibbs distribution,
for which significantly more efficient approximation algorithms have been developed by leveraging samples from the corresponding Gibbs distributions.
Classic examples include volume estimation~\cite{dyer_random_1989}, approximation of the permanent~\cite{jerrumPolynomialtimeApproximationAlgorithm2004}, and approximately counting combinatorial objects such as matchings, independent sets, and spin configurations~\cite{jerrum_polynomial-time_1993,lubyApproximatelyCountingFour1997}, or constraint-satisfying solutions~\cite{fengFastSamplingCounting2021}.

Formally, given a real-valued function $H(\cdot)$ over a finite set $\Omega$, and a weight function $F: \Omega \rightarrow \mathbb R_{\geq 0}$,  the {\em Gibbs distribution} is the family of distributions $\mu^{\Omega}_{\beta}$ over $\Omega$,
parameterized by $\beta$ over an interval $[\betamin, \betamax]$, of the form
$$
\mu^{\Omega}_\beta(\omega)=\frac{\ee^{\beta H(\omega)} F(\omega)}{Z(\beta)}.
$$

The normalization factor $Z(\beta)=\sum_{\omega \in\Omega} F(\omega) \ee^{\beta H(\omega)}$ is called the {\em partition function}. These distributions appear in a number of sampling algorithms and are also common in physics, where the parameter $-\beta$ corresponds to the inverse temperature, and $H(\omega)$ is called the {\em Hamiltonian} of the system.  In the \emph{unweighted} setting, we have $F(\omega) = 1$ for all $\omega$.
Given boundary values $[\betamin, \betamax]$, a key parameter is the partition ratio $Q = Z(\betamax) / Z(\beta_{\min})$.

Since our algorithms only use the sampled value of the Hamiltonian, it is convenient to aggregate states with the same value of $H$.
For each $x \in \mathbb R$, we set $c_x = \sum_{\omega \in H^{-1}(x)} F(\omega)$.
The induced distribution of $X=H(\omega)$, where $\omega$ is drawn from the Gibbs distribution at parameter $\beta$, is the associated \emph{gross Gibbs distribution}.
We denote it by $\mu_{\beta}$; explicitly, 
$$\mu_{\beta}(x) = c_x \ee^{\beta x}/Z(\beta), \quad\text{and} \quad Z(\beta) = \sum_x c_x \ee^{\beta x}.
$$

For $\beta=-\infty$, we let $\mu_{-\infty}$ be the point mass at $0$ and $Z(-\infty)=c_0$.

In the abstract gross-distribution model, the coefficients $c_x$ can only be known up to scaling.  The precise details of the underlying functions $F, H$ and the state space $\Omega$ become essentially irrelevant.

There are two key parameters that will be used to estimate the complexity of our algorithms. We assume  (after rescaling if necessary)  that we are given known values  $q,h  \geq 2$ with
$$
\log Q \leq q, \qquad H(\Omega) \subseteq \{0 \} \cup [1,\infty), \qquad \mathbb{E}\!_{\omega \sim \mu_{\betamax}} \left[H(\omega)\right]\le h.
$$
Here, $q$ is a known finite upper bound on the logarithm of the ratio, and $h$ is a known bound on the expectations. Often, we can simply take $h = \max \{x : c_x > 0 \}$.  In the unweighted setting, when $\betamax \leq 0$, we can simply take $q = \log |\Omega|$.

\medskip

\subsection{Estimating the Partition Ratio}
Many counting problems can be reformulated in terms of the Gibbs partition ratio $Q$ for chosen values $\betamin$ and $\betamax$.  Our goal is to develop \emph{black-box} algorithms to estimate $Q$ using oracle access to samples from the gross distribution $\mu_{\beta}$ at query points $\beta \in [\betamin, \betamax]$.  This approach is often referred to as \emph{simulated annealing}, reflecting the process of gradually varying the parameter $\beta$ from $\betamin$ to $\betamax$.  

 As a concrete example, consider the problem of counting independent sets in an $n$-vertex graph $G$. We take $\Omega = \mathcal I(G)$ to be the set of independent sets of $G$, with $F(I)=1, H(I)=|I|$ for each $I$. The partition function becomes
$Z(\beta)=\sum_{I\in\mathcal I(G)}\ee^{\beta |I|}$.
Taking $\betamax=0$ gives $Z(0)=|\mathcal I(G)|$, while taking $\betamin=-\infty$ gives the trivial value $Z(-\infty)=1$, since only the empty  set contributes.
So estimating $Q=Z(0)/Z(-\infty)$ is exactly independent-set counting.
Here one may take $h=n$ and $q=n\log 2$.

 We define the \emph{sample complexity} to be the number of oracle queries. Early non-adaptive simulated annealing algorithms, due to \cite{dyer_random_1989, Bezakova08}, achieve sample complexity $O(q^2 \eps^{-2} \log h)$ for estimating $Q$ to relative error $\eps$.
A long line of work has since focused on improving this complexity.
The work of~\cite{Stefankovic:JACM09} introduced the first adaptive simulated annealing algorithm, achieving a near-linear sample complexity $O(q \eps^{-2} \polylog(q,h))$.
Subsequent improvements were based on the Tootsie Pop Algorithm (TPA) of~\cite{tpa,TPA:journal} and the paired product estimator (PPE) introduced by~\cite{Huber:Gibbs}.
This line of work~\cite{Huber:Gibbs,Kolmogorov:COLT18,harrisParameterEstimationGibbs2024} culminated in the best known upper bound $O(q \eps^{-2} \log h)$ on sample complexity, achieved in~\cite{harrisParameterEstimationGibbs2024}.
In this oracle model, \cite{harrisParameterEstimationGibbs2024} established a near-matching lower bound of $\Omega(q \eps^{-2})$.\footnote{If the domain is further assumed to be integer-valued, i.e. $H(\Omega)$ is supported on $\{0, 1, \dots, h \}$, then better upper and lower bounds can be shown in some cases; see \cite{harrisParameterEstimationGibbs2024} for further details.}

From the parallel perspective, standard non-adaptive annealing procedures are easy to parallelize since all oracle queries are fixed in advance, but they typically use $O(q^2\eps^{-2}\log h)$ samples.  By contrast, near-optimal annealing algorithms, such as TPA-based methods, are highly adaptive: later query points depend on earlier samples. Hence existing reductions are either easily parallelizable or sample-efficient, but not both.

Our first contribution is a simple non-adaptive algorithm with near-optimal sample complexity.
\begin{theorem}
  \label{mainth1}
  There is a non-adaptive sampling algorithm to estimate $Q$ to relative error $\eps$ with success probability at least $0.7$ in $O(q \eps^{-2} \log^2 h)$ sample complexity.
\end{theorem}

We further show that, with just \emph{two} additional rounds of adaptivity, we can match the best serial algorithms.
\begin{theorem}
  \label{mainth2}
  There is an algorithm using three rounds of sampling to estimate $Q$ to relative error $\eps$ with success probability at least $0.7$  in $O(q \eps^{-2} \log h)$ expected sample complexity.
\end{theorem}

\Cref{tab:sampling-bounds} compares prior work with our results in terms of sample complexity and parallelism in the black-box model.\footnote{Here $\tilde{O}(\cdot)$ suppresses $\polylog(q,h)$ factors. Note that the algorithm of \cite{Stefankovic:JACM09} is presented sequentially with depth $\tilde{O}(q)$, but can be batched to depth $\tilde{O}(q^{1 / 2})$.}
\begin{table}[H]
  \caption{Comparison of Sample Bounds and Parallelism}
  \label{tab:sampling-bounds}
  \centering
  \small
  \setlength{\tabcolsep}{4pt}
  \renewcommand{\arraystretch}{1.08}
  \begin{tabular}{ccc}
    \toprule
    Result & Samples & Rounds of sampling \\
    \midrule
    \cite{dyer_random_1989,Bezakova08} & $O(q^2 \eps^{-2} \log h)$ & $1$ \\
    \cite{Stefankovic:JACM09} & $O(q \eps^{-2}\polylog(q,h))$ & $\tilde{O}(\sqrt{q})$ \\
    \cite{Huber:Gibbs,Kolmogorov:COLT18,harrisParameterEstimationGibbs2024} & $O(q \eps^{-2}\log h)$ & $\tilde{O}(q)$ \\
    This paper & $O(q \eps^{-2} \log^2 h)$ & $1$ \\
    This paper & $O(q \eps^{-2} \log h)$ & $3$ \\
    \bottomrule
  \end{tabular}

\end{table}

We emphasize that all of these black-box algorithms, including the new algorithms we develop, use each sample $\omega$ from the Gibbs distribution only through its Hamiltonian value $H(\omega)$. 

As usual in estimation problems, the success probability can be augmented to any desired value $1-\delta$ by repeating $O(\log(1/\delta))$ times (in parallel) and taking the median. Also, by standard arguments, it suffices to use an approximate sampling  oracle; specifically, if the algorithm has sample complexity $N$, then we can use approximate samples within $o(1 / N)$ total variation distance of the Gibbs distribution.

\subsection{Applications to Parallel Counting}

\Cref{mainth1,mainth2} can be implemented in work-efficient $\RNC$ algorithms in a straightforward way.
By combining with $\RNC$ approximate sampling algorithms, we obtain $\RNC$ algorithms for approximate counting  for several fundamental classes of Gibbs distributions, including:
\begin{itemize}[leftmargin=*]
  \item anti-ferromagnetic $2$-spin systems (including the anti-ferro Ising model and the hard-core model) up to the uniqueness threshold, beyond which approximate counting and sampling become intractable even for sequential algorithms;
  \item the monomer-dimer model (matchings);
  \item the ferromagnetic Ising model with external fields.
\end{itemize}

The following theorems state our results on parallel approximate counting in slightly simplified form.
The formal definitions of the models are postponed to \Cref{section:applications}.
Throughout, we let $G=(V,E)$ be a connected graph with $n$ vertices, $m$ edges, and maximum degree $\Delta$.

\begin{theorem}
  \label{thm:anti-ferro}
  Consider the anti-ferromagnetic $2$-spin system on graph $G$ satisfying one of the following:
  \begin{itemize}[leftmargin=*, labelsep=0.5em, itemsep=0.2em]
    \item strictly anti-ferromagnetic and satisfying the uniqueness condition;
    \item satisfying the uniqueness condition, and $G$ is $\Delta$-regular.
  \end{itemize}
  Then there exists an algorithm that estimates the partition function within relative error $\eps$ with success probability at least $0.7$, running in $O(\log^2 (n/\eps))$ depth using $\tilde{O}(m^2\eps^{-2})$ processors.

  Moreover, for the hard-core model satisfying the uniqueness condition, there exists an algorithm that estimates the partition function within relative error $\eps$ with success probability at least $0.7$, running in $O(\log^2 (n/\eps))$ depth using $\tilde{O}(m n \eps^{-2})$ processors.
\end{theorem}

\begin{theorem}
  \label{thm:monomer-dimer}
  For the monomer-dimer model on graph $G$ with any constant edge weight $\lambda>0$, there exists an algorithm that estimates the partition function within relative error $\eps$ with success probability at least $0.7$,
  running in $O(\Delta^4 \log^3 n \cdot \log (n / \eps))$ depth using $\tilde{O}(m^2  \Delta \eps^{-2})$~processors.
\end{theorem}

\begin{theorem}
  \label{thm:ferro-ising}
  For the ferromagnetic Ising model on graph $G$ with nonzero external fields,
  there exists an algorithm that estimates the partition function
  within relative error $\eps$ with success probability at least $0.7$,
  running in $\polylog(n/\eps)$ depth using $\tilde{O}(nm^2 \eps^{-2})$ processors.
\end{theorem}

Note that the uniqueness condition in \Cref{thm:anti-ferro} matches that in~\cite{chenOptimalMixingTwostate2022};
beyond this regime, approximate counting becomes intractable~\cite{sly2012computational}.

\subsection{Roadmap}
In \Cref{section:PPE}, we recall the paired product estimator and prove the fully non-adaptive result by building a deterministic static schedule whose small max-width implies small curvature.
The three-round algorithm is developed in \Cref{sec:tpa}: there we introduce \texttt{PseudoTPA}, a parallel analogue of the classical TPA schedule-generation process, and show that it provides the width bounds needed to match the best sequential sample complexity.
Finally, \Cref{section:applications} instantiates the black-box reductions for anti-ferromagnetic two-spin systems, the monomer-dimer model, and the ferromagnetic Ising model; the appendices supply the corresponding parallel samplers used in these applications.

In \Cref{fact-appendix}, we collect a few known facts about the partition ratio function $Z(\beta)$. This analysis is standard and has appeared in previous works, e.g. \cite{Huber:Gibbs}; however, since these use different notations and may have slightly different parameters, we include proofs for completeness.

Throughout, we define $z(\beta) = \log Z(\beta)$ and $z(\beta_1, \beta_2) = z(\beta_2) - z(\beta_1)$.

\section{The Paired Product Estimator and the Non-adaptive Schedule}
\label{section:PPE}
In this section, we recap the Paired Product Estimator (PPE) and prove Theorem~\ref{mainth1} by constructing a non-adaptive schedule with low curvature.
The PPE algorithm, introduced in \cite{Huber:Gibbs}, is an estimator for the ratio $Z(\betamax) / Z(\betamin)$.
It is based on a data structure called the \emph{cooling schedule} (we call it just a \emph{schedule} for brevity). Formally, a schedule is a set of values $B = \{ \beta_0 = \betamin, \beta_1, \beta_2, \dots, \beta_t = \betamax \}$ with $\beta_0 < \beta_1 < \dots < \beta_t$. The \emph{length} of the schedule is $\text{len}(B) = t+1$, the \emph{max-width} is $$
\maxwidth(B) = \max_{i=0, \dots, t-1} z(\beta_{i},\beta_{i+1}),
$$
and the \emph{curvature} is defined as
$$
\kappa(B) = \sum_{i=0}^{t-1} \kappa(\beta_i, \beta_{i+1}) \qquad \text{where $\kappa(\alpha_1, \alpha_2) := z(\alpha_1) - 2 z (\frac{\alpha_1 + \alpha_2}{2}) + z(\alpha_2)$}.
$$

Given a schedule $B$, the PPE algorithm shown in \Cref{alg:PPE-wrapper} takes as input a parameter $k$, which dictates the number of samples to use. The algorithm will draw $2 k$ non-adaptive samples from each value $\beta_0, \dots, \beta_t$ (slightly fewer for the endpoints).\footnote{If $\betamin = -\infty$, this algorithm is modified slightly: we set $U_0 = 1$ and $V_0 = \tfrac{1}{k} \sum_{j=1}^k \mathbf 1_{Y_{0,j} = 0}$. We adopt the convention that $\kappa(-\infty, \alpha) = z(-\infty, \alpha)$.  The analysis is very similar in this case; we omit it for brevity.}

\begin{algorithm}
  \For{$i=0,\ldots,t-1$} {
    Draw $k$ samples $X_{i,1}, \dots, X_{i,k}$ from $\mu_{\beta_{i}}$; \\
    Draw $k$ samples $Y_{i,1}, \dots, Y_{i,k}$ from $\mu_{\beta_{i+1}}$; \\
    Set $U_i = \frac{1}{k} \sum_{j=1}^{k} \exp \bigl( \frac{\beta_{i+1}-\beta_{i}}{2} \cdot X_{i,j} \bigr)$ \\
    Set $V_i = \frac{1}{k}\sum_{j=1}^{k}\exp \bigl( -\frac{\beta_{i+1}-\beta_{i}}{2} \cdot Y_{i,j} \bigr)$ \\
  }
  \textbf{return} estimate $\hat Q = \frac{U_0 U_1 \dots U_{t-1}}{V_0 V_1 \dots V_{t-1}}$.
  \caption{\label{alg:PPE-wrapper} Algorithm ${\tt PPE}(B,k)$ for schedule $B$ \\ \texttt{Output:} an estimate of the ratio $Q = Z(\betamax) / Z(\betamin)$}
\end{algorithm}

\begin{theorem}[\cite{Huber:Gibbs}]
  If $\eps \in (0, 1 / 2)$ and $k \geq 100 (\ee^{\kappa(B)} - 1) / \eps^2$, then the estimate $\hat Q$ produced by ${\tt PPE}(B, k)$ is within relative error $\eps$ of $Q$ with probability at least $0.8$.
\end{theorem}
\begin{proof}
  For $i = 0, \dots, t-1$, define $m_i = \tfrac{1}{2} (\beta_i+\beta_{i+1})$ and $\kappa_i = \kappa(\beta_i,\beta_{i+1})$,
  and, for each index $j$,
  $$
  U_{i,j} = \exp\!\left( \frac{\beta_{i+1}-\beta_i}{2} \cdot X_{i,j}\right),
  \quad
  V_{i,j} = \exp\!\left( -\frac{\beta_{i+1}-\beta_i}{2} \cdot Y_{i,j}\right).
  $$
  From \Cref{expprop} and some calculations (see \cite{Huber:Gibbs}), we have:
  $$
  \E[ U_i ] = \E[ U_{i,j} ] = \frac{Z(m_i)}{Z(\beta_{i})}, \quad \E[ V_i ] = \E[ V_{i,j}  ] = \frac{Z(m_i)}{Z(\beta_{i+1})},
  $$
  $$
  \frac{\Var[ U_{i,j} ]}{\E[ U_{i,j} ]^2} =  \frac{ \Var[V_{i,j}]}{\E[V_{i,j}]^2} =  \ee^{\kappa_i} - 1, \quad
  \frac{\Var[ U_i ]}{\E[ U_i ]^2} = \frac{ \Var[V_i]}{\E[V_i]^2} =  (\ee^{\kappa_i} - 1) / k.
  $$

  Let $U = U_0 \cdots U_{t-1}$ and $V = V_0 \cdots V_{t-1}$. Since the values $U_i, V_i$ are all independent, we have
  $$
  \E[ U ] = \prod_{i=0}^{t-1} \frac{Z(m_i)}{Z(\beta_{i})}, \quad \E[ V  ] = \prod_{i=0}^{t-1} \frac{Z(m_i)}{Z(\beta_{i+1})}.$$

  In particular, by telescoping products, we have:
  \begin{align*}
    \frac{\E[U]}{\E[V]} &= \prod_{i=0}^{t-1} \Bigl(  \frac{Z(m_i)}{Z(\beta_{i})} \cdot \frac{Z(\beta_{i+1})}{Z(m_i)} \Bigr) = \prod_{i=0}^{t-1} \frac{Z(\beta_{i+1})}{Z(\beta_{i})}
    = \frac{Z(\beta_t)}{Z(\beta_0)} = \frac{Z(\betamax)}{Z(\betamin)} = Q.
  \end{align*}

  Also, by the formula for products of relative variances,
  \begin{equation}
    \label{bvv1}
    \frac{\Var[ U ]}{\E[U]^2} = \frac{ \Var[V] }{\E[V]^2} =  -1 + \prod_{i=0}^{t-1} (1 + (\ee^{\kappa_i} -1)/k ).
  \end{equation}

  By the fact that $\prod_i(1+(\ee^{\kappa_i}-1)/k) \leq 1+(\ee^{\sum \kappa_i}-1)/k$
  , and since $\sum_{i=0}^{t-1} \kappa_i = \kappa(B)$, the RHS of \Cref{bvv1} is at most $(\ee^{\kappa(B)} -1)/k$. By specification of $k$, this is at most $\eps^2 / 100$. So by Chebyshev's inequality, the bounds $|U / \E[U] - 1| \leq \eps / 3$ and $|V / \E[V] - 1| \leq \eps / 3$ each hold with probability at least $0.9$. When these both hold, the estimate $\hat Q$ satisfies $(1 - \eps) Q \leq \hat Q \leq (1 + \eps) Q$. By the union bound, this occurs with probability at least $0.8$.
\end{proof}

\subsection{Finding a Schedule with Small Curvature}
To use the PPE algorithm, we need an appropriate schedule.  For this, we introduce some definitions.

  For $x \in (\betamin, \betamax)$ and a schedule $B = (\beta_0, \dots, \beta_t)$, there is a unique index $v \in \{0, \dots, t-1 \}$ with $x \in [\beta_v, \beta_{v+1})$. We define $B_-(x) = \beta_v$, $B_+(x) = \beta_{v+1}$ and $W(B,x) = z( B_-(x), B_+(x))$, the ``width'' of the interval containing $x$.

The following result is implicit in \cite{harrisParameterEstimationGibbs2024,Kolmogorov:COLT18} to bound the curvature $\kappa$ for a randomly-generated schedule.
\begin{theorem}
  \label{kappa-thm}
  Let $\theta \in (0,1]$. Suppose we have a random procedure to generate a schedule $B$ such that $\E[W(B, x)] \leq \theta$ for all $x \in (\betamin, \betamax)$.
  Then $\E[ \kappa(B) ] \leq 5 \theta \log( h/\theta)$.
\end{theorem}
\begin{proof}
  First, suppose that $z'(\alpha) = \theta/2$ for $\alpha \in (\betamin, \betamax)$. 
  Let $\beta_v = B_-(\alpha)$. By \Cref{kappaa1}, we calculate
  \begin{align*}
    \kappa(B) &= \sum_{i=0}^v \kappa(\beta_i, \beta_{i+1})  + \sum_{i=v+1}^{t-1} \kappa(\beta_{i}, \beta_{i+1} ) 
    \leq \sum_{i=0}^v z(\beta_i, \beta_{i+1})  + \sum_{i=v+1}^{t-1} z(\beta_{i}, \beta_{i+1} ) \log  \frac{z'(\beta_{i+1})}{z'(\beta_i)}.
  \end{align*}

 The first sum here telescopes to $z(\betamin, \beta_{v+1}) \leq z(\betamin, \alpha) + W(B,\alpha)$. The second sum can be expanded as:
  \begin{align*}
    \sum_{i=v+1}^{t-1} z(\beta_{i}, \beta_{i+1} ) \log \frac{z'(\beta_{i+1})}{z'(\beta_i)}   
    &= \sum_{i=v+1}^{t-1}  \int_{y=\beta_i}^{\beta_{i+1}} W(B, y) \frac{z''(y)  \ \mathrm{d}y}{z'(y)} \leq \int_{y=\alpha}^{\betamax} W(B, y) \frac{z''(y)  \ \mathrm{d}y}{z'(y)}.
  \end{align*}

  Using linearity of expectations along with our hypothesis on $\E[W(B,x)]$ for each $x$, this implies:
  \begin{align*}
    \E[ \kappa(B) ] &\leq z(\betamin, \alpha) + \E[ W(B, \alpha) ] + \int_{y=\alpha}^{\betamax} \E[W(B, y)] \frac{z''(y)  \ \mathrm{d}y}{z'(y)} \\
    &\leq z(\betamin, \alpha) + \theta + \theta \int_{y=\alpha}^{\betamax} \frac{z''(y) \ \mathrm{d}y}{z'(y)} 
    = z(\betamin, \alpha) + \theta \left[ 1 + \log \frac{z'(\betamax)}{z'(\alpha)} \right].
  \end{align*}

  By \Cref{zprimprop1}, we have $z(\betamin, \alpha) \leq 2 z'(\alpha) = \theta$. So overall, $\E[ \kappa(B) ] \leq \theta \bigl( 2 + \log(z'(\betamax)/z'(\alpha)) \bigr)$. Here $z'(\betamax) \leq h$ and  $2 + \log(h/(\theta/2)) \leq  5 \log (h/\theta)$.

  If $z'(\betamin) \ge \theta / 2$, we have $$\kappa(B) \le \int_{y = \betamin}^{\betamax} W(B, y) \frac{z''(y)  \ \mathrm{d}y}{z'(y)},$$
  and the above expectation argument again  gives the desired bound. Finally, if $z'(\betamax) \le \theta / 2$, then $\kappa(B) \le z(\betamin, \betamax) \le 2 z'(\betamax) \le \theta$.
\end{proof}

\begin{corollary}
  \label{kappa-cor}
  For any schedule $B$ with $\maxwidth(B) \leq 1$, it holds that
  $$
  \kappa(B) \leq 5 \maxwidth(B) \log( h/\maxwidth(B)).
  $$
\end{corollary}
\begin{proof}
  Clearly $W(B,x) \leq \maxwidth(B)$ for all $x$. So  \Cref{kappa-thm} holds with $\theta = \maxwidth(B)$.
\end{proof}

\subsection{The Non-adaptive Schedule}
\label{sec:nonadapt}

The original algorithm for generating a non-adaptive schedule appeared in \cite{Bezakova08}. Here, we present a simplified and generalized version of it. It will also be used later in our adaptive algorithm.

\begin{algorithm}[htbp]
  \DontPrintSemicolon
  Set $\beta_0 = \betamax$; \\
  \For {$i=0$ {\bf to} $+\infty$}
  {
    Set $s_i = h$ if $i = 0$, otherwise set $s_i = \min \{h, \frac{q}{\betamax - \beta_i} \}$; \\
    Set $\beta_{i+1} = \beta_i - \theta/s_i$; \\
    \If{$s_i < \theta/2$ or $\beta_{i+1} < \betamin$} {
      \textbf{return} $B = \{ \betamin, \beta_i, \beta_{i-1}, \dots, \beta_{1}, \beta_0 = \betamax \}$.
    }
  }
  \caption{${\tt StaticSchedule}(\theta)$ \\ {\tt Output:} a schedule over interval $[\betamin,\betamax]$ \label{alg:non1}}
\end{algorithm}

Note that by \Cref{derivprop}, each iteration $i$ of \Cref{alg:non1} satisfies  $z'(\beta_i) \leq s_i$. The following theorem shows that the resulting schedule has bounded length and max-width.

\begin{theorem}
  \label{non-adapt}
  The schedule $B = {\tt StaticSchedule}(\theta)$ satisfies  $\len(B) \leq O( \frac{q \log(h/\theta)}{\theta})$ and $\maxwidth(B) \leq \theta$.
\end{theorem}
\begin{proof}
  Let us first show the bound on the max-width.
  Since $z'$ is an increasing function, we have $z(\beta_{i+1}, \beta_{i}) \leq (\beta_{i} - \beta_{i+1}) z'(\beta_i) = (\theta/s_i) \cdot z'(\beta_i)$. This is at most $\theta$ by \Cref{derivprop}.
  If the algorithm terminates due to the condition $\beta_{i+1} < \betamin$, this implies also that $z(\betamin, \beta_i) \leq \theta$. Otherwise, if the algorithm terminates due to the condition $s_i < \theta / 2$, then by \Cref{zprimprop1}, this implies that $z(-\infty, \beta_{i}) \leq 2 z'(\beta_i) \leq 2 s_i \leq \theta$ as desired.

  Next, we bound the schedule length. For the initial iterations of the algorithm, where $\beta_i \geq \betamax - q/h$, we have $s_i = h$ and $\beta_{i+1} = \beta_i - \theta/h$. There are at most $\frac{q/h}{\theta/h} = q/\theta$ iterations with this property.  Subsequently, let $t_1 = \lceil q/\theta \rceil$ denote the first iteration where $\beta_i \leq \betamax - q/h$. For the remaining iterations $i > t_1$ of this algorithm, we have $s_i = \frac{q}{\betamax - \beta_i}$. Thus, after this point, the following recurrence holds:
  \begin{align*}
    s_{i+1} &= \frac{q}{\betamax - \beta_{i+1}} = \frac{q}{\betamax - \beta_i + \theta/s_i}
    = \frac{q}{\betamax - \beta_i + \theta/(q/(\betamax - \beta_i))} = s_i \cdot \frac{q}{q + \theta}.
  \end{align*}

  Thus, for $i > t_1$, we have $s_{i} \leq h (q / (q + \theta))^{i-t_1}$. So when the iteration index $i$ exceeds $t_1 + (\log h - \log(\theta/2)) / (\log (q+\theta) - \log q)$, we have $s_i < \theta/2$ and the algorithm terminates. As $q \geq 2$ and $\theta \in (0,1]$, this is $O(q \log(h/\theta) / \theta)$.
\end{proof}

\begin{proposition}
  \Cref{alg:non1} can be implemented in $\NC$ using $\tilde O(q \log(h/\theta) / \theta)$ processors.
\end{proposition}
\begin{proof}
  We can divide the schedule into two halves: the first part with $s_i = h$, and the second part with $s_i = \frac{q}{\bmax - \beta_i}$.

  In the first half, we have $\beta_i = \betamax - i \theta/h$, and this occurs for iterations $i = 0, \dots, t_1 = \lceil q/\theta \rceil$.
  In the second half, as discussed earlier, $s_i = s_{t_1} ( \frac{q}{q+\theta} )^{i - t_1}$ and $\beta_{i} = \bmax - (\bmax - \beta_{t_1}) (1 + \theta/q)^{i - t_1}$.

  Given these formulas, we first compute $t_1$. We can fill in all the $\beta_i$ values within each half of the schedule, and determine the minimum index with $s_i < \theta/2$ or $\beta_{i+1} < \betamin$.
\end{proof}

\subsection{Putting Everything Together}
We now describe the overall algorithm flow, both generating the schedule and using it for PPE.
First, \Cref{kappa-cor} and \Cref{non-adapt} together immediately give \Cref{mainth1}.

\begin{proof}[Proof of \Cref{mainth1}]
  We generate a schedule by running $B = {\tt StaticSchedule}(1 / (5 \log h))$. Hence $\len(B) \leq O( q \log^2 h)$ and $\kappa(B) \leq 5 \log(5 h \log h) / (5 \log h) \leq 3$. We then run ${\tt PPE}(B,k)$ with $k = 100 (\ee^{3}-1) / \eps^2$. This gives an estimate of $Q$ within relative error $\eps$ with probability at least $0.8$. The total sample complexity is $\len(B) \cdot 2k = O(q \log^2 h / \eps^2)$.
\end{proof}

To get the three-round algorithm, we develop a parallel sampling procedure which is inspired by the TPA process.  The full details are provided later in \Cref{sec:tpa}; we summarize it as follows:
\begin{theorem}
  \label{prune-alg}
  The algorithm ${\tt PseudoTPA}(\theta)$ takes as input a parameter $\theta \in (0,1]$, and produces a random schedule $B$, with the following properties:
  \begin{itemize}
    \item The algorithm has expected sample complexity   $O(q/\theta + q \log h)$, over two rounds of adaptive sampling.
    \item The expected length of the schedule $B$ is $O(q/\theta)$.
    \item For any $x \in (\betamin, \betamax)$, it holds that $\E[ W(B, x) ] \leq \theta$.
  \end{itemize}
\end{theorem}

\begin{proof}[Proof of \Cref{mainth2}, given \Cref{prune-alg}]
  We first generate a schedule $B$ by running ${\tt PseudoTPA}$ with parameter $1 / (5 \log h)$.
  We then estimate $Q$ by running ${\tt PPE}(B,k)$ with $k = 100 (\ee^{30}-1) / \eps^2$.
  By \Cref{kappa-thm}, it holds that $\E[ \kappa(B) ] \leq 5 \log(5 h \log h) / (5 \log h) \leq 3$ and $\E[ \len(B) ] \leq O(q \log h)$.
  The sample complexity of {\tt PseudoTPA} is $O(q \log h)$ in expectation and the expected sample complexity of {\tt PPE} is $\E[ \len(B) \cdot k]= O(q \log h / \eps^2)$.

  Since $\E[\kappa(B)] \leq 3$, Markov’s inequality gives $\kappa(B) \leq 30$  with probability at least $0.9$. When this holds, the estimate produced by {\tt PPE} is good with probability at least $0.8$. Overall, the estimate is accurate with probability at least $0.7$.
\end{proof}

\section{The TPA Algorithm and Its Parallel Counterpart}
\label{sec:tpa}

In this section, we prove \Cref{prune-alg}.
We write $\mathcal E$ for the unit-rate exponential distribution.
To begin, we recall the \emph{TPA algorithm} as used in \cite{Huber:Gibbs}. It is based on an observation of \cite{tpa} (which is targeted for even more general statistical estimation tasks, beyond just Gibbs distributions):

\begin{proposition}[\cite{tpa}]
  \label{exponentialprop}
  Let $\beta \in \mathbb R$. Suppose $X \sim \mu_{\beta}$ and $\eta \sim \mathcal E$. Then $\Pr{ z(\beta - \eta/X, \beta) > t} = \ee^{-t}$ for any $0 \le t < z(-\infty, \beta)$.
  (For $X = 0$ we define $\beta - \eta/X = -\infty$).
\end{proposition}

\begin{proof}
  Fix $0 \le t < z(-\infty, \beta)$.
  Let $\alpha$ be such that $z(\alpha, \beta) = t$. Then
  $$
  \Pr{z(\beta - \eta/X, \beta) > t}
  = \Pr{ \beta - \eta/X < \alpha }
  = \Pr{ \eta > (\beta - \alpha) X}.
  $$
  By definition of $\eta$, we have $\Pr{ \eta > (\beta-\alpha) X \mid X} = \ee^{-(\beta - \alpha) X}$. By iterated expectations along with \Cref{expprop}, this implies that
  \begin{equation*}
  \Pr{ \eta > (\beta-\alpha) X }
  = \E[ \ee^{-(\beta-\alpha) X}]
  = Z(\alpha)/Z(\beta)
  = \ee^{-z(\alpha, \beta)} = \ee^{-t}. \qedhere
  \end{equation*}
\end{proof}

To put it more clearly, $z(\beta-\eta/X,\beta)$ has the distribution of
$\min\{\xi,z(-\infty,\beta)\}$, where $\xi\sim\mathcal E$. Thus the process has
exponential increments in the $z$-coordinate until hitting the boundary, where the final jump is truncated. This memoryless behavior is the key property
underlying the TPA schedule of \cite{Huber:Gibbs}, which is defined as follows:
\begin{algorithm}[htbp]
  \DontPrintSemicolon
  Set $\beta_0 = \betamax$; \\
  \For {$i=0$ {\bf to} $+\infty$}
	{
	  Draw $X$ from $\mu_{\beta_{i}}$  and draw $\eta$ from $\mathcal E$; \\
	  Set $\beta_{i+1} = \beta_{i} - \eta / X$; ($\beta_{i+1} = -\infty$ if $X = 0$) \\
	  \textbf{if  $\beta_{i+1} \leq \betamin$ then return} $B= \{ \betamin, \beta_{i}, \beta_{i-1}, \dots, \beta_0 = \betamax \}$.
	}
  \caption{ {The TPA process \\ {\tt Output:} a schedule over interval $[\betamin,\betamax]$}}
  \label{alg:TPAone0}
\end{algorithm}

\begin{lemma}[\cite{Huber:Gibbs}]
  Suppose we run $k$ independent executions of the TPA Process, resulting in schedules $B_1, \dots, B_k$. Then the combined schedule $B = B_1 \cup \dots \cup B_k$  has expected length at most $k q + 2$, and satisfies \Cref{kappa-thm} with parameter $\theta = 2/k$.
  (This schedule is called the ${\tt TPA}(k)$ process.)
\end{lemma}

\begin{proof}
  By \Cref{exponentialprop}, the distances between successive values $z(\beta_i)$ in each schedule $B_j$ are exponential until the process reaches the boundary. By the superposition principle, the union $z(B \cap (\betamin, \betamax))$ forms a rate-$k$ PPP. The expected number of interior points is at most $k q$, plus we add two boundary points.

  For any $x \in (\betamin, \betamax)$, the waiting time between $z(x)$ and the next-largest value $z(B_+(x))$, with possible truncation at the boundary value $z(x,\betamax)$, is dominated by an exponential random variable of rate $k$. In a completely symmetric way, the waiting time to the next-smallest value $z(B_-(x))$ is also dominated by an exponential random variable of rate $k$. So both $z(x, B_+(x))$ and $z(B_-(x), x)$ have expected value at most $1/k$. 
  This implies that $\E[ W(B, x) ] \leq 2/k$ for all $x \in (\betamin, \betamax)$.
  Therefore, the schedule $B$ satisfies \Cref{kappa-thm} with parameter $\theta = 2/k$.
\end{proof}

Unfortunately, the TPA process is inherently sequential: each value $\beta_{i+1}$ is determined by a sample from the previously chosen distribution $\mu_{\beta_i}$. We do not know how to implement or simulate the TPA algorithm efficiently in the PRAM setting. This motivates \Cref{alg2}, which can be viewed as a local approximation to the ${\tt TPA}(k)$ process.

\begin{algorithm}[H]
  Set parameters $d = 2, \theta' = 1/4, k = \lceil 10/\theta \rceil$; \\
  Set $B = \{\betamin = \beta_0, \dots, \beta_t = \betamax \} = {\tt StaticSchedule}(\theta')$; \\
  Initialize $B' = B'' = \{ \betamin, \betamax \}$ \\
  \For{each $i = 1, \dots, t-1$} {
    Draw $X_{i,1}, \dots, X_{i,d} \sim \mu_{\beta_{i+1}}$; \label{line:alg2-prune-sampling} \\
    Place $\beta_{i}$ into schedule $B'$ with probability $1 -  \ee^{-(X_{i,1} + \dots + X_{i,d}) (\beta_{i+1} - \beta_{i})}$; \\
  }
  \For{each $\beta_i \in B'$}{
    Draw $\eta_{i,1}, \dots, \eta_{i,k} \sim \mathcal E$ and draw $Y_{i,1}, \dots, Y_{i,k} \sim \mu_{\beta_i}$; \label{line:alg2-refine-sampling} \\
    Add $\beta_i$ and $\beta_i - \eta_{i,j}/Y_{i,j}$ for $j = 1, \dots, k$ into $B''$; \\
  }
  \textbf{return} schedule $B'' \cap [\betamin, \betamax]$.
  \caption{\label{alg2} Algorithm ${\tt PseudoTPA}(\theta)$ \\ {\tt Output:} a schedule over interval $[\betamin,\betamax]$}
\end{algorithm}

\begin{proposition}
  \label{fr1prop}
  For any $i \in \{1, \dots, t-1 \}$, we have $\beta_i \in B'$ with probability $1 - \ee^{-d \cdot z(\beta_i, \beta_{i+1})}$.
\end{proposition}
\begin{proof}
  Taking expectation gives
  \begin{align*}
    \Pr{\beta_i  \in B' } &= 1 - \E \bigl[\ee^{-(X_{i,1} + \dots + X_{i,d}) (\beta_{i+1} - \beta_i)} \bigr] 
    = 1 - \prod_{j=1}^d \E\bigl[\ee^{-X_{i,j} (\beta_{i+1} - \beta_{i})}\bigr].
  \end{align*}
  Now apply \Cref{expprop} and the proposition follows.
\end{proof}

\begin{proposition}
  \label{lenprop1}
  $\E[ \len(B')] \leq d q + 2$.
\end{proposition}
\begin{proof}
  Observe that $\sum_{i=1}^{t-1} \Pr{\beta_i \in B'} = \sum_{i=1}^{t-1} (1-\ee^{-d \cdot z(\beta_i, \beta_{i+1})})$.
  The inequality $1-\ee^{-a} \leq a$ gives $\E[\len(B')] \leq 2 + \sum_{i=1}^{t-1} d \cdot z(\beta_i, \beta_{i+1})$.
  By telescoping sums, this equals $2 + d \cdot z(\beta_1, \beta_t) \leq 2 + d q$.
\end{proof}

In light of \Cref{lenprop1}, the schedule $B''$ has expected length at most $\E[ \len(B') \cdot (k+1) ]= O( k q )$ (note that $d$ is constant). The algorithm uses two rounds of sampling: Line 5 uses $d \len(B) = O( q \log h)$ samples and Line 8 uses $k \len(B') = O (k q )$ samples in expectation.

The key task is to bound the curvature of the schedule $B''$. We do this in two stages. First, we show that the interval widths of the schedule $B'$ are tightly concentrated around $O(1)$;  second, we show that $B''$ subdivides each such interval into roughly $k$ equal subintervals.

For any $x \in (\betamin, \betamax)$, define
\[
\begin{aligned}
W'_+(x)&:=z\bigl(x,B'_+(x)\bigr),&
W'_-(x)&:=z\bigl(B'_-(x),x\bigr),&
W'(x)&:=W'_+(x)+W'_-(x)=W(B',x),\\
W''_+(x)&:=z\bigl(x,B''_+(x)\bigr),&
W''_-(x)&:=z\bigl(B''_-(x),x\bigr),&
W''(x)&:=W''_+(x)+W''_-(x)=W(B'',x).
\end{aligned}
\]

\begin{proposition}
  \label{trab1}
  Let $x \in (\betamin, \betamax)$ and $s > 0$. We have
  \[
    \Pr{W'_+(x) > s + \theta'} \leq \ee^{-d s},
    \qquad
    \Pr{ W'_-(x) > s + \theta'} \leq \ee^{-d s}.
  \]
\end{proposition}
\begin{proof}
  Let us first show the bound on $W'_{+}(x)$.  Let $\beta_i = B_+(x)$. If $z(x,\betamax) \leq s + \theta'$, then since $\betamax \in B'$, we have
  $W'_+(x) \leq s + \theta'$ with probability one.
  Thus we may assume that $z(x,\betamax) > s + \theta'$.
  Let $v < t$ be maximal such that $z(x,\beta_v) \leq s + \theta'$.
  Now $W'_+(x) > s + \theta'$ only if
  $B' \cap \{ \beta_i, \beta_{i+1}, \dots, \beta_v \} = \emptyset$.
  By \Cref{fr1prop}, and using independence of the iterations,
  \begin{align*}
    \Pr{ B' \cap \{ \beta_i, \beta_{i+1}, \dots, \beta_v \} = \emptyset }
    &= \prod_{j=i}^{v} \Pr{ \beta_j \notin B'} 
    = \prod_{j=i}^{v}  \ee^{-d \cdot z(\beta_j, \beta_{j+1})} = \ee^{-d \cdot z(\beta_i, \beta_{v+1})}.
  \end{align*}
  
  By definition of $v$, we have $z(x, \beta_{v+1}) > s + \theta'$. By specification of \texttt{StaticSchedule}, we have $z(x, \beta_i) \leq W(B, x) \leq \theta'$. So $z(\beta_i, \beta_{v+1}) > s$
  and the claim follows.

  We now show the bound on $W'_-(x)$. It is immediate if $z(\betamin, x) \leq s+\theta'$. Otherwise, let $y \in (\betamin, x)$  be such that $z(y,x) = s + \theta'$. We have $W'_-(x)  > s + \theta'$ only if $W'_+(y)  > s + \theta'$. By the bound on $W'_+$ above, this has probability at most $\ee^{-d s}$.  
\end{proof}

\begin{proposition}
  \label{trab2}
  Let $x \in (\betamin, \betamax)$. Conditioned on any fixed realization of the schedule $B'$,  we have
  \begin{align*}
    \Pr{W''_+(x) > s} &\leq \bigl( 1 - \ee^{-W'_+(x)}(\ee^{s} - 1) \bigr)^k, \qquad \text{for any $s \in (0, W'_+(x))$} \\
    \Pr{W''_-(x) > s} &\leq \bigl( 1 - \ee^{-W'_+(x)}(1 - \ee^{-s}) \bigr)^k \qquad \text{for any $s \in (0, W'_-(x))$} 
  \end{align*}
\end{proposition}
\begin{proof}
  Let $i$ be the index with $\beta_{i} = B'_+(x)$ and define $r = z(x, \beta_{i}) = W'_+(x)$. For each $j = 1, \dots, k$, define $\alpha_j = \beta_{i} - \eta_{i, j}/Y_{i,j}$ and $V_j = z(\alpha_j, \beta_{i})$. 
  
  By \Cref{exponentialprop}, each $V_j$ has the distribution of a unit-rate exponential random variable truncated at the boundary, and the variables $V_1,\dots,V_k$ are independent.   Since each $\alpha_j$ is placed into $B''$, we have $W''_+(x)  > s$ only if $[z(x),z(x)+s] \cap \{ z(\alpha_1), \dots, z(\alpha_k) \} = \emptyset$. Subtracting from $z(\beta_{i})$, this implies that $[r - s, r] \cap \{ V_1, \dots, V_k \} = \emptyset$.  Since $s < r$, have
  \begin{align*}
    \Pr{ [r - s, r] \cap \{ V_1, \dots, V_k \} = \emptyset }
    &\leq \bigl(1 - (\ee^{-(r-s)}-\ee^{-r})\bigr)^k 
    = \bigl(1 - \ee^{-r}(\ee^{s} - 1)\bigr)^k,
  \end{align*}
  which establishes the first claim.

  In a similar way, $W''_- (x) > s$ only if $[z(x)-s,z(x)] \cap \{ z(\alpha_1), \dots, z(\alpha_k) \} = \emptyset$ which is equivalent to $[r, r + s] \cap \{ V_1, \dots, V_k \} = \emptyset$. Again the truncated exponential law gives 
  \begin{align*}
    \Pr{ [r, r + s] \cap \{ V_1, \dots, V_k \} = \emptyset }
    &\leq \bigl(1 - (\ee^{-r}-\ee^{-r-s})\bigr)^k 
    = \bigl(1 - \ee^{-r}(1 - \ee^{-s})\bigr)^k. \qedhere
  \end{align*}
\end{proof}

\begin{proposition}
  \label{trar2}
  Let $x \in (\betamin, \betamax)$.   Conditioned on any fixed realization of the schedule $B'$, we have
  $\E[ W''_+(x) ] \leq \ee^{W'_+(x)}/k$ and $\E[ W''_-(x) ] \leq \ee^{W'(x)}/k$.
\end{proposition}
\begin{proof}
  Let $r =  W'_+(x)$ and $p = W'_-(x)$.  Observe that $W''_+(x) \leq W'_+(x) = r$ with probability one, since $B' \subseteq B''$. Thus, by \Cref{trab2}, we can write
  \begin{align*}
    \E[ W''_+(x) ] &= \int_{s=0}^{r} \Pr{ W''_+(x) > s } \ \mathrm{d}s 
    \leq \int_{s=0}^{r}  (1 - \ee^{-r}(\ee^s - 1))^k \ \mathrm{d}s.
  \end{align*}

  Using the bound $1+a \leq \ee^a$ twice, we have $1 - \ee^{-r} (\ee^s - 1) \leq \ee^{-s \ee^{-r}}$.
  Plugging this back into our estimate establishes the first claim:
  $$
  \E[ W''_+(x) ] \leq \int_{s=0}^{\infty} ( \ee^{-s \ee^{-r}} )^k \ \mathrm{d}s = \frac{\ee^{r}}{k}.
  $$

  Similarly, $W''_-(x) \leq W'_-(x) = p$ with probability one, so \Cref{trab2} gives:
  \begin{align*}
    \E[ W''_-(x) ] &= \int_{s=0}^{p} \Pr{ W''_-(x) > s} \mathrm{d}s 
    \leq \int_{s=0}^{p}  (1 - \ee^{-r}(1-\ee^{-s}))^k \ \mathrm{d}s.
  \end{align*}

  Since $s \in [0,p]$, we have $1 - \ee^{-s} \geq s \ee^{-p}$ and hence $1 - \ee^{-r} (1 - \ee^{-s}) \leq \ee^{-s \ee^{-r - p}}$.   Plugging this into the bound for $W''_-(x)$, we conclude that
  \[
  \E[ W''_-(x) ] \leq \int_{s=0}^{\infty} \ee^{-k s \ee^{-r-p}} \ \mathrm{d}s = \frac{\ee^{p+r}}{k}. \qedhere
  \]
\end{proof}

Putting these calculations together gives the following:
\begin{proposition}
  For any $x \in (\betamin, \betamax)$, it holds that
  $$\E[W(B'',x)] \leq 10/k \leq \theta.$$
\end{proposition}
\begin{proof}
  By \Cref{trab1}, we can calculate
  \begin{align*}
    \E[ \ee^{W'_+(x)} ] &\leq  \ee^{\theta'}  \Bigl( 1 + \int_{s=0}^{\infty} \Pr{ W'_+(x) > \theta' + s} \ee^s \ \mathrm{d}s \Bigr) \leq \ee^{\theta'} \Bigl( 1 +  \int_{s=0}^{\infty} \ee^{-2 s} \cdot  \ee^s \ \mathrm{d}s \Bigr) = 2 \ee^{\theta'}.
  \end{align*}

  In a completely symmetric way, we have $\E[\ee^{W'_-(x)}]  \leq 2 \ee^{\theta'}$. Since $W'_+(x)$ and $W'_-(x)$ are determined by independent variables, this implies that $\E[ \ee^{W'(x)} ] = \E[ \ee^{W'_+(x)}]  \E[ \ee^{W'_-(x)}] \leq 4 \ee^{2 \theta'}$. For $\theta' = 1/4$, we thus calculate $\E[\ee^{W'_+(x)}] \leq 2.6$ and $\E[\ee^{W'(x)}] \leq 6.6$.

  Plugging these estimates into \Cref{trar2} and using iterated expectations, we have
  \[
    \E[W''_+(x)] \leq \E[ \ee^{W'_+(x)} ]/k \leq 2.6/k,
    \qquad
    \E[W''_-(x)] \leq \E[ \ee^{W'(x)} ]/k \leq 6.6/k.
  \]

  So $\E[W(B'', x)] = \E[W''_-(x)] + \E[ W''_+(x)]  \leq 9.2/k$. For $k = \lceil 10/\theta \rceil$, this is at most $\theta$.
\end{proof}

This concludes the proof of \Cref{prune-alg}.

\section{Applications}

\label{section:applications}
In this section, we apply the framework to obtain a variety of $\RNC$ sampling algorithms. Via the framework of \Cref{mainth2}, these directly lead to efficient $\RNC$ counting algorithms. Note that as we have discussed earlier, we will only generate  approximate samples.

Throughout, we consider an undirected graph $G = (V, E)$ with $n$ vertices and $m$ edges and maximum degree $\Delta$.

\subsection{Anti-ferromagnetic \texorpdfstring{$2$}{2}-spin Systems}

The $2$-spin system $(\gamma_1, \gamma_2, \lambda)$ on $G$ over configurations $\sigma \in \{-1, +1 \}^V$ is defined as
$$
\mu^{\text{2-spin}}_{\gamma_1, \gamma_2, \lambda}(\sigma) \propto \lambda^{n_+(\sigma)} \gamma_1^{m_+(\sigma)} \gamma_2^{m_-(\sigma)},
$$
where $$
n_+(\sigma) = \abs{\set{v \in V \mid \sigma_v = +1}} \qquad m_{\pm}(\sigma) = \abs{\set{(u, v) \in E \mid \sigma_u = \sigma_v = \pm 1}},
$$
and
the parameters $\gamma_1, \gamma_2, \lambda$ satisfy $
0 \le \gamma_1 \le \gamma_2 \neq 0$ and $\lambda > 0.$

The counting problem is to compute the partition function  defined by:
$$Z^{\text{2-spin}}_{\gamma_1, \gamma_2, \lambda} = \sum_{\sigma \in \set{-1, +1}^V} \lambda^{n_+(\sigma)} \gamma_1^{m_+(\sigma)} \gamma_2^{m_-(\sigma)}.$$

The system is called \emph{ferromagnetic} if $\gamma_1 \gamma_2 > 1$, \emph{strictly anti-ferromagnetic} if $\gamma_1 < 1, \gamma_2 < 1$, and \emph{general anti-ferromagnetic} if $\gamma_1  \gamma_2 < 1$. 

\paragraph{Remark on asymptotics.} Throughout this section, we will assume that $\gamma_1, \gamma_2$ are constants; e.g. the Ising model  is a $2$-spin system with $\gamma_1 = \gamma_2$ and the hard-core model has parameters $\gamma_1 = 0, \gamma_2 = 1$.  However, it is critical for our annealing algorithms that $\lambda$ is allowed to vary --- it is \emph{not} treated as a constant. All  asymptotic notations may have a hidden dependence on $\gamma_1$ or $\gamma_2$ (but \emph{not} $\lambda$).

\paragraph{The uniqueness conditions.} For any integer $d \ge 1$, the degree-$d$ univariate tree recursion is defined by $T_d(x) := \lambda \bigl(\frac{\gamma_1 x + 1}{x + \gamma_2} \bigr)^d$; we let $\hat{x}_d$ be the unique positive fixed point of $T_d$, i.e., $T_d(\hat{x}_d) = \hat{x}_d$. The $2$-spin system  is called $d$-unique if $|T_d'(\hat{x}_d)| < 1$, with gap $\delta = 1 -  |T_d'(\hat{x}_d)| \in (0,1)$. 
This is a key stability property for many sampling algorithms, including ours. 
Moreover, in certain systems, failure of the $d$-uniqueness
condition is known to imply worst-case hardness; see,  e.g. \cite{sly2012computational}.

We focus on two cases where the uniqueness condition is particularly simple. First, when $\gamma_2\le 1$, the quantity $|T_d'(\hat x_d)|$ is monotone increasing in $d$, so $(\Delta-1)$-uniqueness automatically implies $d$-uniqueness for all $d < \Delta$. Second, in $\Delta$-regular graphs, many local-tree quantities are governed by
the single branching factor $d = \Delta-1$, and so $(\Delta-1)$-uniqueness is sufficient for the analysis.

So let us fix parameters $0 \le \gamma_1 \leq \gamma_2, \gamma_1 \gamma_2 < 1, \delta \in (0,1)$. The regime for $(\Delta-1)$-uniqueness with gap $\geq \delta$ can be characterized as follows:

\begin{itemize}
\item if $\gamma_1 = 0$, the  regime is $\lambda \in (0, \lambda_c]$ for $\lambda_c = \frac{(1-\delta) \gamma_2^{\Delta} (\Delta - 1)^{\Delta - 1}}{ (\Delta - 2 + \delta)^{\Delta}}$.
  \item otherwise, if $\sqrt{\gamma_1 \gamma_2} \geq \frac{\Delta - 2 + \delta}{\Delta - \delta}$, the  regime is $\lambda \in (0, +\infty)$;
  \item otherwise,  the regime is $\lambda \in (0, \lambda_1] \cup [\lambda_2, +\infty)$, where $x_1 \le x_2$ are the two positive solutions of the equation $\frac{(\Delta-1)(1 - \gamma_1 \gamma_2) x}{(\gamma_1 x + 1)(x + \gamma_2)} = 1 - \delta$ and $\lambda_i := x_i \bigl(\frac{x_i + \gamma_2}{\gamma_1 x_i + 1}\bigr)^{\Delta-1}$ for $i = 1,2$.
\end{itemize}

For these cases, we have crisp results for sampling and counting.

\begin{theorem}[Formal version of \Cref{thm:anti-ferro}]
\label{thm:anti-ferro-stronger}
For a general anti-ferromagnetic $2$-spin system with parameters $(\gamma_1, \gamma_2, \lambda)$  satisfying one of the following conditions:
  \begin{itemize}
    \item when $\gamma_2 \le 1$, $(\gamma_1, \gamma_2, \lambda)$ is $(\Delta - 1)$-unique with gap $\delta$;
    \item when $\gamma_2 > 1$, $(\gamma_1, \gamma_2, \lambda)$ is $(\Delta - 1)$-unique with gap $\delta$ and $G$ is $\Delta$-regular;
  \end{itemize}
  there is an $\RNC$ algorithm to estimate the partition function within relative error $\eps$ with probability at least $0.7$, running in $O_{\delta}(\log \Delta \cdot \log (n / \eps))$ depth using $\tilde{O}_{\delta}(m^2 / \eps^2)$ processors.  
  
  In the hard-core model ($\gamma_1 = 0, \gamma_2 = 1$), the processor count reduces to $\tilde O_{\delta}( m n/\eps^2)$ with the same depth.
  
\end{theorem}

\begin{lemma}
  \label{lem:anti-ferro-sampler}
For an anti-ferromagnetic $2$-spin system with parameters $(\gamma_1, \gamma_2, \lambda)$ satisfying one of the following conditions:
  \begin{itemize}
    \item when $\gamma_2 \le 1$, $(\gamma_1, \gamma_2, \lambda)$ is $(\Delta - 1)$-unique with gap $\delta$;
    \item when $\gamma_2 > 1$, $(\gamma_1, \gamma_2, \lambda)$ is $(\Delta - 1)$-unique with gap $\delta$ and $G$ is $\Delta$-regular;
  \end{itemize}
  there is an $\RNC$ sampler for the Gibbs distribution $\mu^{\text{2-spin}}_{\gamma_1, \gamma_2, \lambda}$ such that an approximate sample within $\eps$ total variation distance can be generated in $O_{\delta}(\log \Delta \cdot \log \frac{n}{\eps})$ depth using $\tilde{O}_{\delta}(m + n)$ processors.
\end{lemma}

 We prove \Cref{thm:anti-ferro-stronger} assuming~\Cref{lem:anti-ferro-sampler} and defer the proof of~\Cref{lem:anti-ferro-sampler} to~\Cref{appendix:anti-ferro-sampler-monomer-dimer}.
 
\begin{proof}[Proof of \Cref{thm:anti-ferro-stronger}]
  By our characterization above, the uniqueness regime with gap $\geq \delta$ may consist of two intervals $(0, \lambda_1]\cup [\lambda_2, +\infty)$, which we call the \emph{lower} and \emph{upper} intervals. If the regime is all of $(0,+\infty)$,
then we regard $(0,1]$
as the lower interval and $[1,+\infty)$ as the upper interval.  It is possible that the upper interval is empty.

By \Cref{lem:anti-ferro-sampler}, there are $\RNC$ samplers for the given Gibbs distributions $\mu^{\text{2-spin}}_{\gamma_1, \gamma_2, \lambda}$ which cover sampling parameters $\hat \lambda$ throughout the entire lower and upper intervals.
  
  For target value $\lambda$ in the lower interval, the annealing process takes a sampling parameter $\hat \lambda = 0$ to $\hat{\lambda} = \lambda$. To fit into our framework, we rewrite the partition function as follows:
  $$Z = \sum_{\sigma \in \Omega} \gamma_1^{m_+(\sigma)} \gamma_2^{m_-(\sigma)} \hat \lambda^{n_+(\sigma)} = \sum_{\sigma \in \Omega} F(\sigma) \cdot \mathrm{e}^{(\ln \hat \lambda) \cdot n_+(\sigma)}.$$
  Here, $\beta = \ln \hat \lambda$ is the inverse temperature, $H(\sigma) = n_+(\sigma)$ is the Hamiltonian, and $F(\sigma) = \gamma_1^{m_+(\sigma)} \gamma_2^{m_-(\sigma)}$ is the weight function.
  We set $\beta_{\min} = -\infty$ and $\beta_{\max} = \ln \lambda$, and take
  $$
  q = m \ln \max\{1,  1/\gamma_2 \} + n\ln(1+\lambda), \qquad h = n,
  $$
  to fit into our framework. Note that $q = O_{\delta}(m)$, and in the hard-core model $q = O_{\delta}(n)$.

  For $\lambda$ in the upper interval, we modify the partition function to fit into our framework as follows:
  $$Z = \hat \lambda^n \cdot \tilde{Z} = \hat \lambda^n \cdot \sum_{\sigma \in \Omega} \gamma_1^{m_+(\sigma)} \gamma_2^{m_-(\sigma)} (1 /\hat \lambda)^{n - n_+(\sigma)}.$$
  Our goal is then to compute the partition function $\tilde{Z}$.
  Here, the inverse temperature is $\beta = \ln (1 / \hat \lambda)$ and the Hamiltonian is $H(\sigma) = n - n_+(\sigma)$. We set $\beta_{\min} = -\infty$ and $\beta_{\max} = \ln (1 / \lambda)$, and set $q = m \ln \max\{ 1/\gamma_1, \gamma_2/\gamma_1 \} + n \ln(1 + 1 / \lambda)$ and $h = n$.   Note again that $q = O_{\delta}(m)$.
  
  The algorithm's depth and the processor count follow by the framework in previous sections and the fact that basic arithmetic operations can be done in polylogarithmic depth.
\end{proof}

\paragraph{Monomer-dimer model.} Let $\mathcal{M}(G)$ denote the set of all matchings $M \subseteq E$ of $G$.     The \emph{monomer-dimer distribution} with constant edge weight $\lambda > 0$ is defined by setting
$$\mu_G(M) \propto \lambda^{|M|} \qquad \text{for $M \in \mathcal M(G)$.}
$$

The partition function of the monomer-dimer model is defined as $$
Z^{\text{MD}}_\lambda = \sum_{M \in \mathcal{M}(G)} \lambda^{|M|},
$$

The monomer-dimer model is also known as the weighted matching distribution. It can be viewed as an anti-ferromagnetic $2$-spin system with parameters $(0, 1, \lambda)$ on the line graph of $G$. Because of the special structure for matchings, the uniqueness regime is not needed for counting or sampling. Any value for the parameter $\lambda > 0$ can be used; we summarize the algorithms as follows.

\begin{theorem}[Formal version of \Cref{thm:monomer-dimer}]
\label{thm:monomer-dimer-stronger}
  For any $\lambda > 0$, there is an $\RNC$ algorithm to estimate the partition function of the monomer-dimer model with edge weight $\lambda$ within relative error $\eps$ with probability at least $0.7$, running in $O_{\lambda}(\Delta^4 \log^3 n \cdot \log (n / \eps))$ depth using $\tilde{O}_{\lambda}(m^2 \Delta / \eps^2)$ processors.
\end{theorem}

\begin{lemma}
  \label{lem:monomer-dimer}
  For the monomer-dimer model with $\lambda > 0$, there is an $\RNC$ sampler within total variation distance $\eps$, using $O_{\lambda}(\Delta^4 \log^3 n \cdot \log (1 / \eps))$ depth and $\tilde{O}_{\lambda}(m \Delta)$ processors.
\end{lemma}

 The proof of \Cref{thm:monomer-dimer-stronger} is identical to that of \Cref{thm:anti-ferro-stronger} by applying \Cref{lem:monomer-dimer}, where the annealing process is from $\hat \lambda = 0$ to $\hat \lambda = \lambda$ and we set $q = O_{\lambda}(m), h = O(m)$. We defer the proof of \Cref{lem:monomer-dimer} to \Cref{appendix:anti-ferro-sampler-monomer-dimer}.

\subsection{Ferromagnetic Ising model}

The \emph{ferromagnetic Ising model with external field} is defined over configurations $\sigma \in \{-1,+1 \}^V$, in terms of a vector of edge activities $\gamma \in [1 + \delta, + \infty)^E$ and a vector of vertex activities $\lambda \in (0,1-\delta]^V$. In this context, we refer to $\delta$ as the \emph{slack}. It assigns each configuration a weight
$$
\mu^{\text{Ising}}_{\gamma, \lambda} (\sigma) \propto \Bigl( \prod_{v: \sigma(v) = +1} \lambda_v \Bigr) \Bigl(  \prod_{\substack{e = (u,v) \in E: \\ \sigma(u) = \sigma(v)}} \gamma_e \Bigr).
$$

We write the corresponding Gibbs distribution and partition function as $\mu^{\text{Ising}}_{\gamma, \lambda},  Z^{\text{Ising}}_{\gamma, \lambda}$. There is a closely related model called the \emph{random cluster model} with parameters $(p, \lambda)$, which is defined by assigning each edge subset $S \subseteq E$ a weight
$$w^{\text{RC}}_{p, \lambda}(S) = \prod_{e \in S} p_e \prod_{e \in E \setminus S} (1 - p_e) \prod_{C \in \kappa(V, S)} (1 + \prod_{v \in C} \lambda_v),$$
where $\kappa(V, S)$ is the set of connected components in the subgraph $(V, S)$. The distribution $\mu^{\text{RC}}_{p, \lambda}$ over $S \subseteq E$ is defined as
$\mu^{\text{RC}}_{p, \lambda}(S) = w^{\text{RC}}_{p, \lambda}(S)/Z^{\text{RC}}_{p, \lambda}$, where $Z^{\text{RC}}_{p, \lambda} = \sum_{S \subseteq E} w^{\text{RC}}_{p, \lambda}(S)$ is the partition function.

It is well-known (see e.g. \cite[Proposition~2.1]{FGW23}) that the ferromagnetic Ising model with edge activities $\gamma$ and vertex activities $\lambda$ is equivalent to the random cluster model with parameters $(p, \lambda)$ such that $p_e = 1 - 1 / \gamma_e$ for each edge $e \in E$ and $Z^{\text{Ising}}_{\gamma, \lambda} = Z^{\text{RC}}_{p, \lambda} \cdot \prod_{e \in E} \gamma_e$.

\begin{theorem}[Formal version of \Cref{thm:ferro-ising}]
\label{thm:ferro-ising-stronger}
  For the ferromagnetic Ising model with slack $\delta$, there is an $\RNC$ algorithm to estimate the partition function within relative error $\eps$ with probability at least $0.7$, running in $O_{\delta} ( \log(n/\eps)^{O(1)})$ depth using $\tilde{O}_{\delta}(nm^2 / \eps^2)$ processors.
\end{theorem}

\begin{lemma}
  \label{lem:ferro-ising-sampler}
  For the ferromagnetic Ising model with slack $\delta$, there is an $\RNC$ sampler for the distribution $\mu^{\text{Ising}}_{\gamma, \lambda}$ within $\eps$ total variation distance in $O_{\delta}(\polylog(n / \eps))$ depth using $\tilde{O}_{\delta}(m^2)$ processors.
\end{lemma}

We prove \Cref{thm:ferro-ising-stronger} assuming \Cref{lem:ferro-ising-sampler}, whose proof is deferred to \Cref{appendix:ferro-ising-sampler}.  Note that a previous parallel sampler is given in \cite{chen2025efficient};  we improve the depth of the algorithm by removing the dependence of the exponent of $\log n$ on the slack $\delta$.

\begin{proof}[Proof of \Cref{thm:ferro-ising-stronger}]
  We rewrite the partition function:
  $$Z(\beta) = \sum_{\sigma \in \Omega} F(\sigma) \exp \Bigl(\beta \cdot \eta \sum_{v: \sigma_v = + 1} \ln (1 / \lambda_v) \Bigr), \qquad \text{where $\eta = -1 / \ln (1 - \delta)$.}
  $$
  
  To fit into our framework, $\Omega = \set{-1, +1}^V$ is the universe, $\beta$ is the inverse temperature, the Hamiltonian function is given by $H(\sigma) = \eta \sum_{v \in V, \sigma_v = + 1} \ln (1/\lambda_v)$ and $F(\sigma) = \prod_{(u, v) \in E, \sigma_u = \sigma_v} \gamma_{u, v}$ is the weight function. The constant $\eta$ is to make sure the Hamiltonian takes value in $\set{0} \cup [1, \infty)$.

  Our target partition function is $Z(-1/\eta)$. It is trivial for $\beta = -\infty$, where $Z(-\infty) = \prod_{(u, v) \in E} \gamma_{u, v}$.
  We set $\beta_{\min} = -\infty$ and $\beta_{\max} = -1/\eta$.
  We take $q = n$ and $h = \max\{2, n \eta/\ee \}$; to justify the latter bound, we calculate:
  \begin{align*}
    \E_{\sigma \sim \mu_{\betamax}} [ H(\sigma) ] &= \sum_{v \in V} \eta \ln(1/\lambda_v) \Pr{ \sigma_v = +1 | \sigma \sim \mu^{\text{Ising}}_{\gamma, \lambda}}
    \end{align*}

    As we show in \Cref{marginal-lemma} in \Cref{appendix:ferro-ising-sampler}, the Ising model satisfies the bound $$
    \Pr{ \sigma_v = +1 | \sigma \sim \mu^{\text{Ising}}_{\gamma, \lambda}} \leq \lambda_v
    $$
    for each $v$ and hence
\begin{align*}
\E_{\sigma \sim \mu_{\betamax}} [ H(\sigma) ] &\leq \sum_{v \in V} \eta \ln(1/\lambda_v) \lambda_v \leq n \eta / \ee
  \end{align*}
  using the elementary inequality $x \log(1/x) \leq 1/\ee$ for any $x \in [0,1]$.

  The algorithm's depth and processor count follow from the framework in previous sections and the fact that all basic arithmetic operations can be done in polylogarithmic depth.
\end{proof}

\appendix

\section{Facts about the Partition Ratio Function for Gibbs Distributions}
\label{fact-appendix}

\begin{fact}
  \label{zprimeprop}
  The derivative of $z$ is given by $z'(\beta) = \mathbb E_{X \sim \mu_{\beta}}[X]$. It is an increasing function in $\beta$.
\end{fact}
\begin{proof}
  The formula $z'(\beta) = \frac{\sum_{x \in \mathbb R} c_x \cdot x \ee^{\beta x}}{Z(\beta)} = \mathbb E_{X \sim \mu_{\beta}}[X]$ can be verified by direct computation. On the other hand, note that $z''(\beta) = {\mathbb V}\text{ar}_{X \sim \mu_\beta}[X]$, which is nonnegative.
\end{proof}

\begin{fact}
  \label{derivprop}
  For any $\beta \in [\betamin, \betamax)$, it holds that $z'(\beta) \leq \min\{h, q/(\betamax - \beta)\}$. Moreover, $z'(\betamax) \leq h$.
\end{fact}
\begin{proof}
  Since $z'(\beta)$ is the expectation of $X \sim \mu_{\beta}$, and by our definition of $h$, we clearly have $z'(\beta) \leq h$. Also, by hypothesis, $q \geq z(\betamax) - z(\betamin)$. Since $\beta \geq \betamin$ this implies that $z(\beta, \betamax) \leq q$. Since $z'$ is increasing, this implies that $z'(\beta) \leq \frac{q}{\betamax - \beta}$ for any $\beta \in [\betamin, \betamax)$.
\end{proof}

\begin{fact}
  \label{expprop}
  For any values $\alpha, \beta$, we have 
  $\E_{X \sim \mu_{\beta}}[\ee^{\alpha X} ] = Z(\alpha + \beta) / Z(\beta).$
\end{fact}
\begin{proof}
  The proof is a straightforward calculation:
  \[
  \E[\ee^{\alpha X}] = \sum \ee^{\alpha x} \cdot \frac{c_x \ee^{\beta x}}{Z(\beta)} = \frac{ \sum_x c_x \ee^{(\alpha + \beta) x} }{Z(\beta)} = \frac{Z(\alpha + \beta)}{Z(\beta)}. \qedhere
  \]
\end{proof}

\begin{fact}
  \label{zprimprop1}
  If $z'(\beta) \leq 1/2$, then $z(-\infty,\beta) \leq 2 z'(\beta)$.
\end{fact}
\begin{proof}
  Here $Z(-\infty) = c_0$, and $Z(\beta) = c_0 + \sum_{x \neq 0} c_x \ee^{\beta x}$.
  We have assumed that $\mu_{\beta}$ is supported on $\{0 \} \cup [1,\infty)$.
  From \Cref{zprimeprop}, this implies that
  $$
  Z(\beta) \leq c_0 + \sum_{x} x c_x \ee^{\beta x}
  = c_0 + z'(\beta) Z(\beta).
  $$

  So $Z(\beta) \leq \frac{c_0}{1 - z'(\beta)}$ and $z(-\infty, \beta) \leq -\log(1-z'(\beta))$; since $z'(\beta) \leq 1/2$, this is at most $2 z'(\beta)$.
\end{proof}

\begin{fact}[\cite{Huber:Gibbs}]
  \label{kappaa1}
  For $-\infty < \beta_1 < \beta_2$ the curvature $\kappa$ as defined in \Cref{section:PPE} satisfies $$\kappa(\beta_1, \beta_2) \leq z(\beta_1, \beta_2) \cdot \min \{1, \log( z'(\beta_2) / z'(\beta_1) ) \}.$$
\end{fact}
\begin{proof}
  Let $\alpha = \frac{\beta_1 + \beta_2}{2}$. So $\kappa(\beta_1, \beta_2) = z(\alpha, \beta_2) - z(\beta_1, \alpha)$. From monotonicity of $z$, this immediately implies $\kappa(\beta_1, \beta_2) \leq z(\alpha, \beta_2) \leq z(\beta_1, \beta_2)$.
  For the second bound, let $v = z'(\beta_2)/z'(\beta_1) \geq 1$. It suffices to show it for $v \leq e$. Since $z'$ is increasing, we have:
  $$
    z(\alpha, \beta_2) \leq z'(\beta_2) (\beta_2 - \alpha), \qquad
    z(\beta_1, \alpha) \geq z'(\beta_1) (\alpha - \beta_1), \qquad
    z(\beta_1, \beta_2) \geq z'(\beta_1) (\beta_2 - \beta_1).
    $$
    
  Thus,
  \begin{align*}
    \frac{ \kappa(\beta_1, \beta_2) }{z(\beta_1, \beta_2)} &\leq \frac{z'(\beta_2) (\beta_2 - \alpha) - z'(\beta_1) (\alpha - \beta_1)}{z'(\beta_1)(\beta_2 - \beta_1)} = \frac{z'(\beta_2) - z'(\beta_1)}{2 z'(\beta_1)}  = (v - 1)/2.
  \end{align*}
  Since $1 \leq v \leq \ee$, we have $v - 1 \leq 2 \log v$.
\end{proof}

\section{$\RNC$ Samplers for the Anti-ferromagnetic \texorpdfstring{$2$}{2}-spin Model}
\label{appendix:anti-ferro-sampler-monomer-dimer}
In this section we prove \Cref{lem:anti-ferro-sampler} and \Cref{lem:monomer-dimer}. We apply the parallel single-site dynamics framework developed in~\cite{liuParallelizeSingleSiteDynamics2024} to give $\RNC$ samplers. The proofs simply combine the optimal mixing time of the sequential Glauber dynamics and a bound on the Dobrushin influence matrix.
The Glauber dynamics is a widely used Markov chain for sampling from Gibbs distributions. In each step, it selects a vertex uniformly at random and updates the spin of the vertex according to the marginal distribution conditioned on the spins of other vertices.

Formally, the mixing time $T_{\text{mix}}(\eps)$ of a Markov chain with stationary distribution $\mu$ is defined as $T_{\text{mix}}(\eps) := \min \set{t \ge 0 \mid \max_{X_0 \in \Omega} \mathrm{d}_{\mathrm{TV}}(X_t, \mu) \le \eps}$,
where $X_0$ is the initial configuration, $X_t$ is the configuration after $t$ steps of the Markov chain.

Here we consider a graph $G = (V,E)$ with maximum degree $\Delta$. We treat $\gamma_1, \gamma_2$ as constants for our asymptotic bounds. The following two theorems give the mixing time of the Glauber dynamics for the scenarios of   \Cref{lem:anti-ferro-sampler} and \Cref{lem:monomer-dimer} respectively.

\begin{theorem}[{\cite[Theorem 1.3]{chenOptimalMixingTwostate2022}}]
  \label{theorem:mixing-time-anti-ferro}
  Let $\delta \in (0, 1)$.
  For the anti-ferromagnetic $2$-spin system with parameters $(\gamma_1, \gamma_2, \lambda)$ satisfying one of the following conditions:
  \begin{itemize}
    \item when $\gamma_2 \le 1$, $(\gamma_1, \gamma_2, \lambda)$ is $(\Delta - 1)$-unique with gap $\delta$;
    \item when $\gamma_2 > 1$, $(\gamma_1, \gamma_2, \lambda)$ is $(\Delta - 1)$-unique with gap $\delta$ and $G$ is $\Delta$-regular;
  \end{itemize}
  The mixing time of the Glauber dynamics for the Gibbs distribution $\mu^{\rm{\textnormal{2-spin}}}_{\gamma_1, \gamma_2, \lambda}$ is bounded as
  $$T_{\textnormal{mix}}(\eps) \le C(\delta) \cdot n \tp{2 \log n + \log \log \alpha + \log \log (\lambda + \lambda^{-1}) + \log (1 / 2 \eps^2)},$$
  where $\alpha$ is a constant depending on $\gamma_1, \gamma_2$ and  $C(\delta) = \exp(O(1 / \delta))$ is a constant depending on $\delta$.
\end{theorem}

\begin{theorem}[{\cite[Theorem 30]{chen2025faster}}]
  \label{theorem:mixing-time-monomer-dimer}
  For the monomer-dimer model with edge weight $\lambda > 0$, the mixing time of the Glauber dynamics for the Gibbs distribution $\mu^{\rm{\textnormal{MD}}}_{\lambda}$ is bounded as
  $$T_{\textnormal{mix}}(\eps) \le O_{\lambda}(\Delta^3 m \log^2 n \cdot \log(1/\varepsilon)).$$
\end{theorem}

We then introduce the Dobrushin influence matrix and its connection to parallel simulation.

\begin{definition}[Dobrushin influence matrix]
For a Gibbs distribution $\mu$ over sample space
$\Omega=\{+1,-1\}^V$, the Dobrushin influence matrix
$R=(R_{u,v})_{u,v\in V}$ is defined by $R_{v,v}=0$ and, for $u\neq v$,
\[
R_{u,v}
=
\max_{\substack{
\sigma,\tau\in\{+1,-1\}^{V\setminus\{v\}}\\
\sigma_{V\setminus\{u,v\}}=\tau_{V\setminus\{u,v\}}
}}
d_{\mathrm{TV}}\!\left(\mu_v^\sigma,\mu_v^\tau\right),
\]
where 
$\mu_v^\sigma$ is the marginal distribution at $v$ conditioned on
configuration $\sigma$ on $V\setminus\{v\}$.
\end{definition}

\begin{theorem}[{\cite[Theorem 1.2]{liuParallelizeSingleSiteDynamics2024}}]
  \label{theorem:parallel-single-site}
  For any single-site dynamics updating vertices on a spin system on $G$ with Dobrushin influence matrix $R$ satisfying $\norm{R}_p \le \rho$ for some $p \ge 1, \rho \ge 1$, there is a parallel algorithm to simulate $T$ steps of the dynamics in $O(\rho \cdot (T / n + \log n) \cdot \log \Delta)$ depth using $\tilde{O}(m + n)$ processors.
\end{theorem}

We show \Cref{lem:anti-ferro-sampler} and \Cref{lem:monomer-dimer} by combining the above theorems with a bound on the Dobrushin influence matrix.

\begin{lemma}[{\cite[Variant of Claim 8.17]{chenRapidMixingGlauber2024}}]
  \label{lem:anti-ferro-influence-matrix}
  For an anti-ferromagnetic $2$-spin system that is $(\deg(u)-1)$ unique at each vertex $u$ with gap $\delta$, the Dobrushin influence matrix $R$ satisfies $\norm{R}_1 \le 64 / \delta^2$.
\end{lemma}

\begin{proof}[Proof of \Cref{lem:anti-ferro-sampler}] Define $A = \max\{1, \gamma_2^{-\Delta} \}$. If $\lambda \leq \eps/(n A)$, then we simply return the fixed all-minus configuration. This deterministic distribution is within $\eps$ of the Gibbs distribution since, for any vertex $v$ with $k_{\pm}$ neighbors of spin $\pm 1$, the conditional odds of assigning spin $+1$ to $v$ are $\lambda\gamma_1^{k_+} \gamma_2^{k_-}
   \leq \lambda A$,  and so by a union bound the total probability of any $+1$ vertex is at most $n \lambda A \leq \eps$. 

A similar argument shows that when 
$\gamma_1 > 0, \lambda\geq n/(\eps\gamma_1^\Delta)$, we can simply return the all-plus configuration,  within $\eps$ total-variation distance of the Gibbs distribution. When $\gamma_1 = 0$, then the fact that $\lambda$ is within the uniqueness regime implies that $\lambda \leq \exp(O_{\delta}(\Delta))$. 

Thus, aside from these aforementioned cases,  we assume that both $\lambda$ and $\lambda^{-1}$ are at most $\exp(O_{\delta}(n/\eps))$. So $\log\log(\lambda+\lambda^{-1})
       = O_{\delta}(\log(n/\eps)).$  \Cref{theorem:mixing-time-anti-ferro} therefore shows that $$
       T= O_{\delta}(n \log(n/\eps))
       $$
       steps of the Glauber dynamics suffice to produce a sample within
total-variation distance $\eps$.

By   \Cref{lem:anti-ferro-influence-matrix} and \Cref{theorem:parallel-single-site}, we can simulate these $T$ steps of Glauber updates in parallel in depth
$O_\delta (\log\Delta\cdot\log(n/\eps) )$ using $\widetilde O_\delta(m+n)$ processors.
\end{proof}

\begin{proof}[Proof of \Cref{lem:monomer-dimer}]
  Consider the Glauber dynamics for the monomer-dimer model with edge weight $\lambda>0$.
  By \Cref{theorem:mixing-time-monomer-dimer}, running the Glauber dynamics for $T = O_{\lambda}(\Delta^3 m \log^2 n \cdot \log(1/\varepsilon))$ steps yields a distribution within total variation distance $\varepsilon$ from the stationary distribution $\mu_{G}$.

  For the monomer-dimer Glauber dynamics, each step updates the state of a single edge. Thus, we will apply \Cref{theorem:parallel-single-site} to the line graph of $G$.    For the line graph, the Dobrushin influence matrix $R \in \mathbb{R}^{m \times m}$ has a nonzero entry $R_{f,e}$  only when $e$ and $f$ share a vertex.
  Since each edge is incident to at most $2(\Delta-1)$ other edges (excluding itself), each column of $R$ has at most $2(\Delta-1)$ nonzero entries. Since each entry $R_{f,e}$ lies in $[0,1]$, we have $\|R\|_1 \le 2\Delta$.
   \Cref{theorem:parallel-single-site} then gives  a parallel speedup factor of $\Omega(m / \Delta)$.
  So the parallel depth required to simulate $T$ steps of the Glauber dynamics is $T \cdot \Delta/m \cdot \log \Delta = O_{\lambda}(\Delta^4 \log^3 n \cdot \log(1/\varepsilon))$, while using  $\tilde{O}(m \Delta)$ processors.
\end{proof}

\section{$\RNC$ Sampler for the Ferromagnetic Ising Model}
\label{appendix:ferro-ising-sampler}

In this section we prove \Cref{lem:ferro-ising-sampler}. By the external-field Edwards--Sokal coupling~\cite[Proposition~2.3]{FGW23} and its parallelism in \cite[Lemma 3.3]{chen2025efficient}, it suffices to give a parallel sampler for the corresponding random cluster model.

Our algorithm uses the \emph{field dynamics}, a novel Markov chain first introduced in~\cite{chenRapidMixingGlauber2024}. Given a parameter $\theta \in (0, 1)$ for the random cluster model $\mu^{\text{RC}}_{p, \lambda}$, this process updates a configuration $X \subseteq E$ as follows:
\begin{itemize}[leftmargin=*]
  \item Add each edge $e \in E$ into $S$ independently with probability $\theta$;
  \item Replace $X$ with a draw from the distribution $\mu^{\text{RC}}_{p^*, \lambda}$, where $p^*_e = \frac{p_e}{p_e + \theta (1 - p_e)} {\mathbf 1}_{e \in S \cup X}$.
\end{itemize}

Here $\mathbf 1$ is the indicator function. Equivalently, $X$ is drawn from the distribution $\mu^{\text{RC}}_{p^*, \lambda}$ with vector $p^* = \frac{p}{p + \theta (1 - p)}$ restricted to subgraph $G' = (V, S \cup X)$.

Note that a previous parallel sampler is given in \cite{chen2025efficient}; we improve the depth of the algorithm by removing the dependence of the exponent of $\log n$ on $\delta$.

\begin{lemma}
\label{marginal-lemma}
  For any vertex $v$, the ferromagnetic Ising model has $\Pr{\sigma_v = +1} \leq \lambda_v$.
\end{lemma}
\begin{proof}
  We can write the marginal probability as
  $$
  \Pr{\sigma_v = +1} = \frac{ \sum_{\sigma: \sigma_v = +1} w(\sigma)}{\sum_{\sigma} w(\sigma) } = 1 - \frac{ \sum_{\sigma: \sigma_v = -1} 
 \prod_{v': \sigma(v') = +1} \lambda_{v'} \prod_{\substack{e = (u',v') \in E: \\ \sigma(u') = \sigma(v')}} \gamma_e}{ Z^{\text{Ising}}_{\gamma, \lambda}}.
  $$

  The sum in the numerator is precisely the partition function $Z^{\text{Ising}}_{\gamma, \tilde \lambda}$ 
  where $\tilde \lambda$ is obtained from $\lambda$ by setting $\tilde \lambda_v = 0$. That is, we have 
  $\Pr{\sigma_v = +1} = 1 - Z^{\text{Ising}}_{\gamma, \tilde \lambda}/Z^{\text{Ising}}_{\gamma, \lambda}$. In turn, for any subset of edges $S \subseteq E$, the ratio of the random-cluster model weights for $\lambda$ and $\tilde \lambda$ can be bounded as
  
  $$
  1 \ge \frac{w^{\text{RC}}_{p, \tilde{\lambda}}(S)}{w^{\text{RC}}_{p, \lambda}(S)} = \frac{1}{1 + \prod_{u \in C_v} \lambda_u} \ge 1 - \prod_{u \in C_v} \lambda_u \ge 1 - \lambda_v,$$
  where $C_v$ is the connected component containing vertex $v$ in the subgraph $(V, S)$.
  Summing over all subsets of edges $S \subseteq E$, we conclude that
  \begin{equation*}
    1 \ge \frac{Z^{\text{Ising}}_{\gamma, \tilde{\lambda}}}{Z^{\text{Ising}}_{\gamma, \lambda}} = \frac{Z^{\text{RC}}_{p, \tilde{\lambda}}}{Z^{\text{RC}}_{p, \lambda}} = \frac{\sum_{S \subseteq E} w^{\text{RC}}_{p, \tilde{\lambda}}(S)}{\sum_{S \subseteq E} w^{\text{RC}}_{p, \lambda}(S)} \ge 1 - \lambda_v. \qedhere
  \end{equation*}
\end{proof}

\begin{lemma}[{\cite[Variant of Theorem 8.1 and Lemma 8.4]{fengRapidMixingGlauber2025}}]
  \label{lem:field-dynamics-mixing}
  The field dynamics started from the all-edge configuration satisfies
  $$T^{\rm{\textnormal{FD}}}_{\text{mix}}(\eps) \le \Bigl( \frac{2}{p_{\min} (1 - \lambda_{\max})} \Bigr)^{30/(1 - \lambda_{\max})^2} \cdot \frac{ \log(n/\eps) }{ \theta^{2 \cdot 10^5} }
  $$
  where $p_{\min} = \min_e p_e$ and $\lambda_{\max} = \max_{v \in V} \lambda_v$.
\end{lemma}
\begin{proof}
A more general precise bound is stated in \cite{fengRapidMixingGlauber2025}. Since our field dynamics starts from the all-edge configuration $E$, the same entropy argument gives
  $$T^{\rm{\textnormal{FD}}}_{\text{mix}}(\eps) \le \ee^{ \int_0^{\log 1/\theta} 4 \alpha(t) \, \d t }\cdot\Bigl(\log \log \frac{1}{\mu^{\rm{\textnormal{RC}}}_{p,\lambda}(E)} + \log \frac{1}{2\eps^2}\Bigr) + 1,$$
  where $\theta_0 := p_{\min} (1 - \lambda_{\max})^2/2$,
  and
  $$\alpha(t) := \begin{cases} 3(1 - \lambda_{\max})^{-2} & \text{if } 0 \le t \le \log \frac{1}{\theta_0}, \\ 5 \cdot 10^4 & \text{if } t > \log \frac{1}{\theta_0}. \end{cases}$$

  We can bound the integral term by:
  \begin{align*}
    \int_0^{\log 1/\theta} 4 \alpha(t) \d t &\leq
    \int_0^{\log 1/\theta_0} 4 \cdot 3 (1 - \lambda_{\max})^{-2} \d t + 2 \cdot 10^5 \log \theta^{-1} \\
    &= 12 ( 1 - \lambda_{\max})^{-2} \log \Bigl( \frac{2}{(1 - \lambda_{\max})^2 p_{\min}} \Bigr) + 2 \cdot 10^5 \log \theta^{-1}.
  \end{align*}

  We use the lower bound $\mu^{\rm{\textnormal{RC}}}_{p,\lambda}(E) \geq p_{\min}^{n^2} / 2^{n^2}$ for the random cluster model, by the same argument as in \cite[Appendix A.2]{chen2025efficient}: the all-edge configuration has weight at least $p_{\min}^m$, while the partition function is at most $2^{n^2}$.
  Hence
  $$T^{\rm{\textnormal{FD}}}_{\text{mix}}(\eps)
    \le \Bigl( \frac{2}{p_{\min}(1 - \lambda_{\max})^2} \Bigr)^{12/(1 - \lambda_{\max})^2}
    \theta^{-2 \cdot 10^5} \cdot
    \Bigl( 2 \log n + \log \log \frac{2}{p_{\min}} + \log \frac{1}{2 \eps^2} \Bigr) + 1.$$
  The stated bound follows by an upper bound on this quantity.
\end{proof}

The basic idea for our sampling algorithm is to simulate the field dynamics with an appropriate parameter $\theta$, where each update of the field dynamics is in turn simulated by the Glauber dynamics. The key is to choose $\theta$ so as to place the random cluster model $\mu^{\mathrm{RC}}_{p^*, \lambda}$ into a more tractable regime, where the Glauber dynamics admits simple analysis for fast mixing and efficient parallelization.

We first recall the parallelism result for the large-$p$ random cluster model in~\cite{chen2025efficient} as follows.

\begin{lemma}[{\cite[Variant of Theorem 4.3]{chen2025efficient}}]
  \label{lem:parallel-glauber-rc}
  Suppose that $n \ge 3$, $\eps \in (0, 1 / 2)$ and
  $$
  p_{\min} \geq 1 - \frac{1 - \lambda_{\max}}{\ee^{45} \log (n/\eps)}.
  $$
  
  The Glauber dynamics for the random cluster model $\mu^{\rm{\textnormal{RC}}}_{p, \lambda}$ is fast mixing and can be efficiently parallelized.
  Formally, there is a parallel algorithm {\rm{\texttt{ParallelRC}}}$(G, p, \lambda, \eps)$ to simulate the Glauber dynamics with
  $O(\log^2 (m/\eps))$ depth and $\tilde{O}(m^2 \log(1/\eps))$ processors,
  such that the output distribution is within total variation distance $\eps$ from $\mu^{\rm{\textnormal{RC}}}_{p, \lambda}$.
\end{lemma}
\begin{proof}
  The algorithm is identical to \cite[Algorithm 2]{chen2025efficient}, except with a change in parameters
  $T^{\textnormal{GD}} = \lceil 2m \log (4m/\eps) \rceil$ and $T^{\textnormal{PA}} = \lceil 3 \log (4 T^{\textnormal{GD}} / \eps) \rceil$.
  The algorithm simulates $T^{\textnormal{GD}}$ steps of the Glauber dynamics within $T^{\textnormal{PA}}$ depth. The proof of \cite[Theorem 4.3]{chen2025efficient} consists of two parts: on one hand, the Glauber dynamics returns a random sample within $\eps / 2$ total variation distance of $\mu_{p, \lambda}^{\text{RC}}$ within $T^{\textnormal{GD}}$ steps; on the other hand, the parallel algorithm successfully simulates the sequential dynamics with probability at least $1 - \eps / 2$. The only change to make this proof work with our new condition on $p_{\min}$ is to change the definition of the large-expansion vertex set families to:
  \begin{align*}\mathcal{C} &:= \set{S \subseteq V \mid \abs{S} \le n / 2 \text{ and } \abs{E(S, V \setminus S)} \ge \abs{S} \log (n / \eps)} \\  
  \mathcal{G} &:= \set{X \subseteq E \mid \forall S \in \mathcal{C}, \abs{X \cap E(S, V \setminus S)} > 0}.
  \end{align*}  
  where $E(S, V \setminus S)$ denotes the set of edges between $S$ and $V \setminus S$.
  
  We claim that, under the given condition on $p_{\min}$, the Condition 1 of \cite{chen2025efficient} holds with $\mathcal{G}$ and $\eta = \eps / (16 T^{\textnormal{GD}})$. The verification is essentially identical to that in \cite{chen2025efficient}, except for the corresponding change in parameters. Conditioned on this good event, the analysis of the parallel algorithm is the same as in the proof of \cite[Theorem 4.3]{chen2025efficient}.
\end{proof}

The sampling algorithm for the random cluster model is now given in~\Cref{alg5}, denoted \texttt{SampleRC}, below.
It takes as input a graph $G = (V, E)$, edge activities $p \in (0, 1)^E$, vertex activities $\lambda \in [0, 1)^V$ and error bound $\eps \in (0, 1)$.

\begin{algorithm}[H]
  \DontPrintSemicolon
  \caption{\label{alg5}\label{alg:random-cluster-sampler} Algorithm \texttt{SampleRC}$(G,p,\lambda,\eps)$}

  \nonl
  \vspace{0.15em}
  \begingroup
  \setlength{\fboxsep}{4pt}
  \noindent\hspace*{-\algomargin}\hspace*{0.25in}%
  \fbox{%
  \begin{minipage}{\dimexpr\linewidth-0.5in-2\fboxsep-2\fboxrule\relax}
  \centering
  {\footnotesize
  \textbf{Parameters}\par
  \vspace{0.2em}
  \begin{tabular}{@{}c@{\hspace{1.25em}}c@{\hspace{1.25em}}c@{}}
  $C_1 = 10^{-27}(1-\lambda_{\max})p_{\min},$
  &
  $C_2 = \bigl(2p_{\min}^{-1}(1-\lambda_{\max})^{-1}\bigr)^{30/(1-\lambda_{\max})^2},$
  &
  $N_0 = \max\{1/C_1,C_2\},$
  \\[0.35em]
  \multicolumn{3}{c}{
    \begin{tabular}{@{}c@{\hspace{3.5em}}c@{}}
    $\theta = C_1/\log(n/\eps),$
    &
    $T = C_2\theta^{-2\cdot 10^5}\log(2n/\eps).$
    \end{tabular}
  }
  \end{tabular}
  }
  \end{minipage}%
  }\;
  \endgroup
  \vspace{-0.3em}
\;

  \If{$n < N_0$}{
    \textbf{return} $X \sim \mu^{\mathrm{RC}}_{p,\lambda}$ by brute force \;
  }
  initialize $X \gets E$ \;
  \For{$i = 1,2,\ldots,T$}{
    construct $S \subseteq E$ by including each edge $e \in E$ independently with probability $\theta$ \;
    update $X \gets \texttt{ParallelRC}(G',p^*,\lambda,\frac{\eps}{2T})$, with
    $G' = (V,S \cup X)$ and
    \[
      p^*_e =
      \frac{p_e}{p_e+\theta(1-p_e)}
      \mathbf{1}[e \in S \cup X]
    \]
    \tcc{Simulate Glauber dynamics for the biased distribution on $G'$.}
  }
  \textbf{return} $X$ \;
\end{algorithm}

The following is the key observation about the algorithm:
\begin{proposition}
  When $n \geq N_0$, the biased distribution $p^*$ satisfies the preconditions of \Cref{lem:parallel-glauber-rc}.
\end{proposition}
\begin{proof}
  For the biased distribution $p^*$, we have
  $$
  1 - p_{\min}^* = 1 - \frac{p_{\min}}{p_{\min} + \theta (1 - p_{\min})}  \leq \frac{1 - \lambda_{\max}}{10^{27} \log(n/\eps)}
  $$

  We claim that this satisfies the condition of \Cref{lem:parallel-glauber-rc} with the parameter $\eps' = \frac{\eps}{2 T}$. To show this, it suffices to show that
  \begin{equation}
    \label{ggr3}
    10^{27} \log(n/\eps) \geq \ee^{45} \log(2 T n / \eps)
  \end{equation}

  To that end, we expand out the definition of $T$ to calculate:
  \begin{align*}
    (\ee^{45} \log(2 T n / \eps)) / (10^{27} \log(n/\eps))
    &\leq \;  10^{-7} \Bigl( 1 + \frac{\log (2 C_2)}{\log(n/\eps)} + \frac{ -2 \cdot 10^5 \log C_1}{\log(n/\eps)} + \frac{(2 \cdot 10^5 + 1) \log(n/\eps)}{\log(n/\eps)} \Bigr) \\
    &\leq \; 10^{-7} \Bigl( 1 + \log 2 + \frac{\log C_2}{\log C_2} + \frac{ -2 \cdot 10^5 \log C_1}{\log(1/C_1)} + (2 \cdot 10^5 + 1) \Bigr),
  \end{align*}
  which is at most $0.05$. This establishes \Cref{ggr3}.
\end{proof}

\begin{lemma}
  The algorithm {\tt SampleRC} uses $\tilde O_{\delta}(m^2)$ processors and depth $O_{\delta} ( \log(n/\eps)^{O(1)})$. It returns a sample within total variation distance $\eps$ from $\mu^{\text{RC}}_{p, \lambda}$.
\end{lemma}
\begin{proof}
  If $n < N_0$, we sample by brute force using $2^{O(N_0^2)} \leq O_{\delta}(1)$ processors. So we suppose $n \geq N_0$.

  The algorithm runs for $T$ iterations and each iteration runs in $O(\log^2(m T /\eps))$ depth using $\tilde{O}(m^2 \log(T/\eps))$ processors. Plugging in the parameters gives the claimed complexity bounds.

  Next, we turn to the output distribution. By \Cref{lem:parallel-glauber-rc}, each iteration of {\tt ParallelRC} samples from the Glauber dynamics, up to variation distance of $\frac{\eps}{2 T}$. Thus, up to variation distance $T \cdot \frac{\eps}{2 T}$, the procedure simulates the $T$ steps of the field dynamics. Also, by \Cref{lem:field-dynamics-mixing}, we have $T \geq T^{\text{FD}}_{\text{mix}}(\eps/2)$. So, the field dynamics itself produces a sample within total variation distance $\eps/2$ of $\mu_{p, \lambda}^{\text{RC}}$.

  Overall, the error incurred is at most $T \cdot \frac{\eps}{2 T} + \frac{\eps}{2} = \eps$, concluding the proof of the sampling algorithm.
\end{proof}
{
\printbibliography
}
\end{document}